%% file: csJournal.tex
\def\IEEEsubmission{0}
\input{variables}

\if\IEEEsubmission1
\documentclass[journal,12pt,onecolumn,draftclsnofoot]{IEEEtran}
\else
\documentclass[journal]{IEEEtran}
\fi
\usepackage{multicol}
\usepackage{lipsum}
\usepackage{mathtools}
\usepackage{amsthm}
\usepackage{acronym}
\usepackage[bookmarksopen=true]{hyperref}
\usepackage{stfloats}
\usepackage{amsfonts}
\usepackage{acronym}
\usepackage[noadjust]{cite}
\usepackage{multirow}
\usepackage{bm}
\usepackage{amsmath,amssymb}
\usepackage{graphicx}
\usepackage{epstopdf}
\usepackage[table,xcdraw]{xcolor}
\usepackage[caption=false,font=footnotesize]{subfig}
\usepackage[geometry]{ifsym}
\usepackage{array}
\usepackage[utf8]{inputenc}
\usepackage[T1]{fontenc}
\usepackage{pifont}
\usepackage{tabularray}
\usepackage{lipsum}
\def\BibTeX{{\rm B\kern-.05em{\sc i\kern-.025em b}\kern-.08em
		T\kern-.1667em\lower.7ex\hbox{E}\kern-.125emX}}

\input{otherDefinitions}

\input{acronyms}
\begin{document}

\title{
	{Reliable  Majority Vote Computation with Complementary Sequences for  UAV Waypoint Flight Control}\\
	\thanks{Alphan~\c{S}ahin and Xiaofeng~Wang are with the Electrical  Engineering Department,
		University of South Carolina, Columbia, SC, USA. E-mail: asahin@mailbox.sc.edu, wangxi@cec.sc.edu}
		\thanks{This paper was presented in part at the IEEE Military Communications Conference 2023  \cite{sahinMILCOMuav_2023}.}	
	\author{Alphan~\c{S}ahin,~\IEEEmembership{Member,~IEEE} and Xiaofeng~Wang,~\IEEEmembership{Member,~IEEE}} 
}
\maketitle

\begin{abstract}
In this study, we propose a non-coherent \acl{OAC} scheme to  calculate the \ac{MV} reliably in fading channels. The proposed approach relies on modulating the amplitude of the elements of \aclp{CS} based on the sign of the parameters to be aggregated. Since it does not use channel state information at the nodes, it is compatible with time-varying channels. To demonstrate the efficacy of our method, we employ it in a scenario where an \acl{UAV} is guided by distributed sensors, relying on the MV computed using our proposed scheme.  We show that the proposed scheme notably reduces the computation error rate  with a longer sequence length in fading channels while maintaining the peak-to-mean-envelope power ratio of the transmitted orthogonal frequency division multiplexing signals to be less than or equal to 3~dB.

\end{abstract}
\section{Introduction}
\acresetall

Multi-user interference is often considered an undesired phenomenon for communication systems as it can degrade the link performance. In contrast, the same underlying phenomenon, i.e., the signal superposition property of wireless multiple-access channels, can be very useful in the computation of special mathematical functions by harnessing the additive nature of the wireless channel. The gain obtained with \ac{OAC} is that the resource usage can be reduced to a one-time cost, which otherwise scales with the number of devices \cite{Altun_2021survey,sahinSurvey2023,Zheng_2023}. Hence, \ac{OAC} can benefit applications  by reducing the latency when a large number of nodes participate in computation over limited resources. Nevertheless, the transmitted signals for \ac{OAC} are perturbed by noise and distorted by fading and hardware impairments such as \ac{PA} non-linearity and synchronization errors. To address these issues, in this work, we propose a non-coherent \ac{OAC} method for \ac{MV} computation with \acp{CS} \cite{Golay_1961}. We demonstrate its efficacy for an \ac{UAV} waypoint flight control scenario,  where a \ac{UAV} is guided by many distributed sensors harnessing the multi-user interference to use the limited wireless resources as efficiently as possible.

\subsection{Related Work and Challenges}
\subsubsection{Over-the-air computation}
The idea of function computation over a \ac{MAC} was first thoroughly analyzed in Bobak's pioneering work in \cite{Nazer_2007}. In \cite{goldenbaum2013harnessing} and \cite{goldenbaum2015nomographic}, Goldenbaum shows that \ac{OAC} can be utilized to compute a family of functions, i.e., nomographic functions, including arithmetic mean, norm, polynomial function, maximum, and \ac{MV}.
OAC  has recently gained momentum with an increased number of applications where the ultimate purpose of  communications is computation. For example, the authors in \cite{chen2021distributed,Sahin_2022MVjournal,Guangxu_2020,Guangxu_2021}  implement \ac{FL} \cite{pmlr-v54-mcmahan17a} over a wireless network, where \ac{OAC} is used for aggregating a large number of gradients or model parameters of the edge devices at an edge server efficiently. Similarly, \ac{OAC} is considered for split learning in \cite{Krouka_2021} to aggregate smashed data, i.e., the outputs of a neural network. 
We refer the readers to \cite{Altun_2021survey,sahinSurvey2023,Zheng_2023,hellstrom2020wireless}, and \cite{Zhibin_2022oac} and the references therein for \ac{OAC} and its exciting applications, such as distributed localization, wireless data centers, and wireless control systems. 

The primary challenge of computing functions via signal superposition is that the receiver observes the superposition of the  signals {\em distorted} by the wireless channels between the receiver and transmitters. To address this issue, a large number of studies adopt pre-equalization techniques,  where the parameters desired to be aggregated are distorted with the reciprocals of the channel coefficients before the transmission so that a coherent superposition is achieved at the receiver \cite{Amiri_2020,Guangxu_2020,Guangxu_2021,Wei_2022}. This approach can provide excellent results when the phase synchronization among the devices can be maintained. Furthermore, as the effective channel corresponds to an \ac{AWGN} channel, the reliability of \ac{OAC} can be improved further with lattice codes \cite{goldenbaum2015nomographic,Nazer_2007,Lan_2023,XieICC_2024}. However, in practice, it is very difficult to maintain the phase synchronization  as the phase response of the {\em composite} wireless channel (i.e., including the responses of the transmitter and receiver) is a strong function of mobility and hardware impairments such as clock errors, residual \ac{CFO}, and time-synchronization errors. For instance,  a single sample deviation can cause large phase rotations in the frequency domain \cite{sahinDemo2022} and only a $\pm15$-degree phase mismatch across the transmitters can degrade the \ac{OAC} performance dramatically \cite[Fig. 4]{XieICC_2024}. In \cite{Haejoon_2021}, it is shown that non-stationary  channel conditions in a \ac{UAV} network can severely deteriorate the coherent signal superposition. 
In \cite{Amiria_2021} and \cite{Busra_2023}, a more practical scheme where the transmitters are blind (i.e.,  no \ac{CSI} at the transmitters) while the receiver has an estimate of the aggregated \ac{CSI} (i.e., the sum of fading coefficients across the links) is considered. It is shown that using a large number of antennas at the receiver can mitigate the interference components arising due to the inner product operation to estimate the superposed values. However, this approach may not be viable when the computing node has limited space and battery life. To overcome the phase synchronization bottleneck, another approach is to use non-coherent OAC  at the expense of sacrificing more resources. For instance, two orthogonal resources  are allocated to compute an \ac{MV} function in \cite{Sahin_2022MVjournal} and \cite{safiMVlocOJCS} by exploiting the energy accumulation via modulation techniques such as \ac{FSK}, \ac{PPM}, and \ac{CSK}. In \cite{sahinGC_2022},  the authors demonstrate this approach in practice by using five Adalm Pluto \acp{SDR} to train a neural network  without phase synchronization across the transmitters. Similarly, in \cite{Akshay_IPSN2020}, orthogonal resources are used  for negative- and positive-valued aggregation with \ac{FSK} and \ac{CSK}.  By using more resources along with a decomposition relying on a balanced number system, a quantized \ac{OAC} is 
investigated in \cite{sahinGCbalanced_2022}. The aforementioned non-coherent techniques exploit type-function, i.e., frequency histogram, by defining discrete classes to compute functions such as arithmetic mean, maximum, minimum, and median \cite{Mergen_2006tsp,Mergen_2007}. 
In \cite{Goldenbaum_2013tcom} and \cite{Goldenbaum_wcl2014}, random unimodular sequences are proposed to be utilized at the devices. In this method, a transmitter modulates the energy of the sequence with the parameter to be aggregated. At the receiver, the energy of the received superposed sequence is calculated  for continuous-valued OAC at the expense of interference components. A proof of concept demonstration regarding Goldenbaum's approach is provided in \cite{Kortke_2014}. Nevertheless,
 achieving reliable computation with a non-coherent OAC technique in a fading channel is still an unsolved issue in the literature.

The second challenge for reliable computation arises because the received signal powers of the nodes need to be similar, if not identical. 
If a transmitted signal for \ac{OAC} has a large \ac{PMEPR}, it can result in a reduced cell size due to the power back-off or a higher adjacent channel interference due to the \ac{PA} saturation \cite{safiMVlocOJCS}. The adjacent channel interference can also increase further due to the simultaneous transmissions from many nodes participating in \ac{OAC}. In the literature, few \ac{OAC} schemes are analyzed from the perspective of \ac{PMEPR}. To reduce \ac{PMEPR}, chirps and \ac{SC} waveforms are used in \cite{safiMVlocOJCS} and \cite{Cai_2018}, respectively. It is well-known that the \ac{PMEPR} of the \ac{OFDM} signals using \acp{CS} is less than or equal to 3~dB while achieving some coding gain \cite{davis_1999}. However, to the best of our knowledge, \acp{CS} have not been utilized for reliable \ac{OAC} while reducing the dynamic range of the transmitted \ac{OFDM} signals.

\subsubsection{Wireless Control Systems}
Suppose a control unit needs to receive feedback from a larger number of  sensors over wireless links for a control application. In this scenario,
without OAC, orthogonal wireless resources must be allocated to the sensors to receive sensory data, and the control unit must wait for the acquisition to be completed to perform the desired computation. As a result, the latency (or resource consumption) increases linearly with the number of sensors, which can cause an unstable system response or a slower system. To address this issue, the stability of a dynamic plant is investigated under limited wireless resources in \cite{Cai_2018}, and \ac{OAC} is exploited to compute the feedback from a large number of distributed sensors as quickly as possible to ensure  the stability of a dynamic plant. In \cite{Park_2021}, a general state-space model of a discrete-time
linear time-invariant system is proposed to be computed with \ac{OAC}. In \cite{Jihoon_2022Platooning}, \ac{OAC} is utilized to achieve mean consensus for a vehicle platooning application. It is shown that all vehicles converge to a specific value proportional to the average position without using an orthogonal multiplexing technique.  With the same motivation for reducing the latency, in this work, we consider a scenario where a UAV receives feedback from distributed erroneous sensors to infer its flying direction as quickly as possible so that it can  fly stably. To our knowledge, the \ac{UAV} waypoint flight control scenario has not been investigated in the literature by taking \ac{OAC} into account. 
It is also worth mentioning that  \ac{OAC} is investigated in specific scenarios that involve \acp{UAV}. For instance, in \cite{Xiang_ojc2022}, the \acp{UAV} compute the arithmetic mean of ground sensor readings with OAC. In \cite{Min_uav2022}, \ac{UAV} trajectories are optimized based on the locations of the sensors. However, the purpose of OAC is not wireless control in these papers.  
%
%

\subsection{Contributions}
In this study, we focus on computing \acp{MV} reliably in fading channels. Our contributions  can be listed as follows:

\begin{itemize}
\item We propose  a new non-coherent \ac{OAC} scheme based on \acp{CS} \cite{sahin_2020gm} to improve the robustness of computation against fading channels while limiting the dynamic range of transmitted signals to mitigate the distortion due to hardware non-linearity. Since the proposed approach does not rely on the availability of \ac{CSI} at the transmitters and receiver, it also provides robustness against time-varying channels and time-synchronization errors. 

\item By extending our preliminary work in \cite{sahinMILCOMuav_2023},  we rigorously analyze the \ac{CER} of the proposed \ac{OAC} scheme. We derive the \ac{CER} in Corollary~\ref{cor:cerpqzKpKnKz} and Corollary~\ref{cor:cerpqz} based on Lemma~\ref{lemma:errProbGivenVote}.

\item We demonstrate the applicability of the proposed method to a \ac{UAV} flight control scenario based on \ac{MV} computation. We provide the corresponding convergence analysis and show that the proposed approach is globally uniformly ultimately bounded in mean square in Theorem~\ref{th:oacMVstable}.

\item We support our findings with comprehensive simulations. We also generate numerical results based on Goldenbaum's \ac{OAC} scheme in \cite{Goldenbaum_2013tcom} to provide a comparative analysis.
\end{itemize}

{\em Organization:} The rest of the paper is organized as follows. Section~\ref{sec:system} provides the notation and  preliminary discussions used in the rest of the sections. 
In Section~\ref{sec:scheme}, the proposed \ac{OAC} scheme is discussed in detail. In Section~\ref{sec:performance}, we theoretically analyze the \ac{CER} of the proposed scheme.
In Section~\ref{sec:convergence}, the convergence of the \ac{UAV} waypoint flight control is discussed. In Section~\ref{sec:numerical}, we assess the proposed scheme numerically. We conclude the paper in Section \ref{sec:conclusion}. A summary of the notation used throughout the paper is given in \tablename~\ref{table:notation}.
\begin{table}[t]
	\caption{Notation summary.}
	\centering
		\resizebox{\if\IEEEsubmission1 6in \else 3.5in \fi}{!}{	
		\begin{tabular}{@{}p{0.625in}|p{ \if\IEEEsubmission1 5in \else 2.75in \fi  }}
			\textbf{Notation} & \textbf{Meaning of the notation}\\
			\hline\hline
			$\complexNumbers$ & The sets of complex numbers \\\hline
			$\realNumbers$ & The sets of real numbers \\\hline
			$\integers_H$ & The sets of integers modulo $H$ \\\hline
			$\integers^\numberOfIterations_\numberOfPointsForPSK$ &
			The set of $\numberOfIterations$-dimensional integers where each element is in $\integers_\numberOfPointsForPSK$\\\hline
			$\expectationOperator[\cdot][x]$ & The expectation of its  argument over all random variables\\\hline
			$\complexGaussian[0][\sigma^2]$ & A zero-mean symmetric complex Gaussian distribution with variance $\sigma^2$ \\ \hline
			$\uniformDistribution[a][b]$ & The uniform distribution with the support between $a$ and $b$ \\\hline
			$\normalCDF[\cdot]$& The \ac{CDF} of the standard normal distribution \\\hline
			$\CDF[x][][a;b]$& The \ac{CDF} of a random variable $x$ evaluated at $a$ for a given parameter $b$ \\\hline
			$\charFcn[t]$ & The characteristic function of a random variable $x$, i.e., $\expectationOperator[\constante^{\constantj t x}][]$ \\\hline
			$\probability[A;x]$ & The probability of  the event $A$ with a parameter $x$\\\hline
			$\probability[A|B]$ & The conditional probability of an event $A$ given the event $B$\\\hline
			 $\oneVector[L]$ & A vector of length $L$, where its elements are only $1$\\\hline
			 $\zeroVector[L]$ & A vector of length $L$, where its elements are only $0$\\\hline
			 $(\eleGa[{\indexEleOfSeq}])_{i=0}^{\lengthGaGb-1}$ & A sequence of length $\lengthGaGb$, i.e., $\seqGa=(\eleGa[0],\eleGa[1],\dots, \eleGa[\lengthGaGb-1])$ \\\hline
			 $x^*$& The complex conjugate of $x\in\complexNumbers$\\\hline
			 $\max\{a,b\}$& The maximum element of $(a,b)$\\\hline
			 $\signNormal[\cdot]$& The signum function\\\hline
			 $\floor{\cdot}$& The floor function\\\hline
			 $\ceil{\cdot}$& The ceiling function\\\hline
			 $\funcfForANF:\integers^\numberOfIterations_2\rightarrow\realNumbers$ & A pseudo-Boolean function\\\hline
			 $\constante$ & Euler's constant\\\hline
			 $\constantj$ & $\sqrt{-1}$\\\hline
		\end{tabular}
	}
	\label{table:notation}
\end{table}

\section{System Model}
\label{sec:system}
Consider a scenario where a \ac{UAV} flies from one point of interest $(\locationInitialEle[1],\locationInitialEle[2],\locationInitialEle[3])$ to another point of interest $(\locationTargetEle[1],\locationTargetEle[2],\locationTargetEle[3])$. Suppose that the \ac{UAV} cannot localize its location in the room. However, it can receive feedback from $\numberOfEdgeDevices\ge1$  sensors\footnote{To localize the UAVs, the sensors can rely on computer vision techniques and use wide-angle cameras. For example, the sensor can exploit the changes in the features of the environment by comparing them with those previously created 3-D maps or use stereo vision capture and depth maps, as discussed in \cite{MUSTAFAH2012575} and \cite{Carrio_2020}.} deployed in the room about the velocity of the \ac{UAV} on the $x$-, $y$-, and $z$-axis for every $\Trefresh$ seconds. Based on the feedback from the sensors, the \ac{UAV} updates its position at the $\indexRound$th round for the $x$-, $y$-, and $z$-axis, denoted by $\locationEle[\indexRound][1]$, $\locationEle[\indexRound][2]$, and $\locationEle[\indexRound][3]$, respectively, as
\begin{align}
	\locationEle[\indexRound+1][\indexCoordinate] &= \locationEle[\indexRound][\indexCoordinate] -\Trefresh\velocityVectorEle[\indexRound][\indexCoordinate]~, \label{eq:dynamics}
\end{align}
where
\begin{align}
	\velocityVectorEle[\indexRound][\indexCoordinate]&= 
	\begin{cases}
		\max\{{\updateRate\feedbackEle[\indexRound][\indexCoordinate]},-\maximumVelocity\}& \feedbackEle[\indexRound][\indexCoordinate]<0\\
		\min\{{\updateRate\feedbackEle[\indexRound][\indexCoordinate]},\maximumVelocity\}& \feedbackEle[\indexRound][\indexCoordinate]\ge0
	\end{cases}
	\label{eq:clamp},
\end{align}
 for $\locationEle[0][\indexCoordinate]\triangleq\locationInitialEle[\indexCoordinate]$, $\forall\indexCoordinate\in\{1,2,3\}$. In \eqref{eq:dynamics} and \eqref{eq:clamp}, $\velocityVectorEle[\indexRound][\indexCoordinate]$ is the velocity at the $\indexRound$th round for the $\indexCoordinate$th coordinate,  $\maximumVelocity>0$ is the maximum velocity of the \ac{UAV}, $\updateRate>0$ is the update rate, $\feedbackEle[\indexRound][\indexCoordinate]$ is the velocity-update strategy given by
\begin{align}
	\feedbackEle[\indexRound][\indexCoordinate]=\begin{cases}
		\displaystyle\frac{1}{\numberOfEdgeDevices}\sum_{\indexED=1}^{\numberOfEdgeDevices}\locationEstimateEle[\indexRound][\indexED,\indexCoordinate]-\locationTargetEle[\indexCoordinate],&\text{Cont. (ideal)}\\
		\underbrace{\signNormal[{\sum_{\indexED=1}^{\numberOfEdgeDevices}{\underbrace{\signNormal[{\locationEstimateEle[\indexRound][\indexED,\indexCoordinate]-\locationTargetEle[\indexCoordinate]}]}_{\triangleq\voteVectorEDEleCoordinate[\indexED,\indexCoordinate][\indexRound]}}}]}_{\triangleq\majorityVoteEle[\indexCoordinate]},
		&\text{MV (ideal)}\\
		\majorityVoteDetectedEle[\indexCoordinate],&\text{MV (OAC)}
	\end{cases},
	\label{eq:strategies}
\end{align}
where 
$\locationEstimateEle[\indexRound][\indexED,\indexCoordinate]=\locationEle[\indexRound][\indexCoordinate]+\locationEstimationErrorEle[\indexRound][\indexED,\indexCoordinate]$ is an estimate of the $\indexCoordinate$th coordinate of the UAV position at the $\indexED$th sensor, $\voteVectorEDEleCoordinate[\indexED,\indexCoordinate][\indexRound]$ is the $\indexED$th sensor's vote for the $\indexCoordinate$th coordinate, $\locationEstimationErrorEle[\indexRound][\indexED,\indexCoordinate]$ is a zero-mean Gaussian variable with the variance $\varianceSensor$, and $\majorityVoteEle[\indexCoordinate]$ and $\majorityVoteDetectedEle[\indexCoordinate]$ denote the
 \ac{MV} computed under perfect communications and the  \ac{MV} obtained with the proposed scheme (i.e., \eqref{eq:decision}) for the $\indexCoordinate$th coordinate, respectively. We use the term {\em ideal} to imply perfect communication between the sensors and the UAV.

In \eqref{eq:strategies}, the UAV averages the sensors’ outputs  to decide which direction to fly for the continuous case. In the MV cases, such  averaging is not possible. Instead, the UAV goes in the direction (with the length $\pm\updateRate\Trefresh$) on each coordinate based on the majority of the sensors’ output. Note that an \ac{MV}-based update was previously investigated in machine learning literature, showing that an MV-based update is related to an update based on the sign of the median value \cite{Chennips_2020}. Finally, we note that the model can be more sophisticated than the one in \eqref{eq:dynamics} and take other dynamics, such as \ac{UAV} imperfections, into account~\cite{bouabdallah2007full}. Since our paper focuses on \ac{OAC}, we use \eqref{eq:dynamics} as a baseline control model to assess the proposed OAC scheme.

\subsection{Complementary Sequences}
Let $\seqGa=(\eleGa[{\indexEleOfSeq}])_{i=0}^{\lengthGaGb-1}\triangleq(\eleGa[0],\eleGa[1],\dots, \eleGa[\lengthGaGb-1])$ be a sequence  of length $\lengthGaGb$ for $\eleGa[{\indexEleOfSeq}]\in\complexNumbers$ and $\eleGa[\lengthGaGb-1]\neq0$. We associate the sequence $\seqGa$ with the polynomial
$\polySeq[\seqGaP][\polyVariable] = \eleGa[\lengthGaGb-1]\polyVariable^{\lengthGaGb-1} + \eleGa[\lengthGaGb-2]\polyVariable^{\lengthGaGb-2}+ \dots + \eleGa[0]$
in indeterminate $\polyVariable$. The \ac{AACF} of the sequence $\seqGa$ given by
\begin{align}
	\apac[\seqGa][\lagForCorrelation]\triangleq
	\begin{cases}
		\sum_{\indexEleOfSeq=0}^{\lengthGaGb-\lagForCorrelation-1} \eleGa[\indexEleOfSeq]^*\eleGa[\indexEleOfSeq+\lagForCorrelation], & 0\le\lagForCorrelation\le\lengthGaGb-1\\
		\sum_{\indexEleOfSeq=0}^{\lengthGaGb+\lagForCorrelation-1} \eleGa[\indexEleOfSeq]\eleGa[\indexEleOfSeq-\lagForCorrelation]^*, & -\lengthGaGb+1\le\lagForCorrelation<0\\
		0,& \text{otherwise}
	\end{cases}~.
\end{align}
If $\apac[\seqGa][\lagForCorrelation]+\apac[\seqGb][\lagForCorrelation] = 0$ holds for $\lagForCorrelation\neq0$, 
the sequences $\seqGa$ and $\seqGb$ are referred to as \acp{CS} \cite{Golay_1961}. It can be shown the \ac{PMEPR} of an \ac{OFDM} symbol constructed based on a \ac{CS} is less than or equal to 3~dB \cite{davis_1999}.

Let $\funcfForANF(\seqx)$ be a map from $\integers^\numberOfIterations_2=\{\seqx\triangleq(\monomial[1],\monomial[2],\dots, \monomial[\numberOfIterations])|\forall \monomial[\indexFirstOrderMonomial]\in\integers_2\}$ to $\realNumbers$ as $\funcfForANF:\integers^\numberOfIterations_2\rightarrow\realNumbers$, i.e., a pseudo-Boolean function.
A family of \acp{CS} can be obtained by using pseudo-Boolean functions  as follows:
\begin{theorem}[{\cite{sahin_2020gm}}]
	\label{th:reduced}
	Let 
	$\seqPermutationCompShift=(\permutationMono[\indexIteration])_{\indexIteration=1}^{\numberOfIterations}$ be a permutation of $\{1,2,\dots,\numberOfIterations\}$. For any $\numberOfPointsForPSK,\numberOfIterations\in\integersPositive$, $\scaleEexp[\indexIteration],\scaleEexp[0]\in\realNumbers$, and $\angleexpAll[\indexIteration],\angleexpAll[0] \in \integers_\numberOfPointsForPSK$ for $\indexIteration\in\{1,2,\mydots,\numberOfIterations\}$, let
	\begin{align}
&\funcfForFinalAmplitude(\seqx)
= \sum_{\indexIteration=1}^{\numberOfIterations}\scaleEexp[\indexIteration]\monomialAmp[{\permutationMono[{\indexIteration}]}]+\scaleEexp[0]\label{eq:realPartReduced}~,
\\		
		&\funcfForFinalPhase(\seqx)
		= {\frac{\numberOfPointsForPSK}{2}\sum_{\indexIteration=1}^{\numberOfIterations-1}\monomial[{\permutationMono[{\indexIteration}]}]\monomial[{\permutationMono[{\indexIteration+1}]}]}+\sum_{\indexIteration=1}^\numberOfIterations \angleexpAll[\indexIteration]\monomial[{\permutationMono[{\indexIteration}]}]+  \angleexpAll[0]\label{eq:imagPartReduced}~,
	\end{align}
where  $
	\monomialAmp[{\permutationMono[{\indexIteration}]}]$ is $
		(\monomial[{\permutationMono[{\indexIteration}]}] +\monomial[{\permutationMono[{\indexIteration+1}]}])_2$ and 	$\monomial[{\permutationMono[{\numberOfIterations}]}]$ for $ \indexIteration<\numberOfIterations$ and $\indexIteration=\numberOfIterations$, respectively. 	Then, the sequence $\transmittedSeqNoRound[]=(\transmittedSeqEleNoRound[0],\mydots,\transmittedSeqEleNoRound[\lengthOfSequence-1])$, where  its associated polynomial  is  given by
	\begin{align}
		\polySeq[{\transmittedSeqP}][\polyVariable] &= 
		\sum_{\forall\seqx\in\integers^\numberOfIterations_2}\underbrace{
		\constante^{\funcfForFinalAmplitude(\seqx)}
		\constante^{\constantj\frac{2\pi}{\numberOfPointsForPSK}\funcfForFinalPhase(\seqx)}  }_{\transmittedSeqEleNoRound[\funcEnum(\seqx)]} 
		\polyVariable^{\funcEnum(\seqx)}~,\label{eq:encodedFOFDMonly}
	\end{align}
	is a \ac{CS} of length $\lengthOfSequence=2^\numberOfIterations$, where $\funcEnum(\seqx) \triangleq   \sum_{\indexFirstOrderMonomial=1}^{\numberOfIterations}\monomial[\indexFirstOrderMonomial]2^{\numberOfIterations-\indexFirstOrderMonomial}$, i.e.,  a decimal representation of the binary number constructed using all elements in the sequence $\seqx$. 
\end{theorem}
Theorem~\ref{th:reduced} shows that the functions that determine the amplitude and the phase of the elements of the CS $\transmittedSeqNoRound[]$ (i.e., $\funcfForFinalAmplitude(\seqx)$ and $\funcfForFinalPhase(\seqx)$) and \ac{RM} codes have similar structures. The function $\funcfForFinalPhase(\seqx)$ is in the form of the cosets of the first-order \ac{RM} code within the second-order \ac{RM} code  \cite{davis_1999}. Notice that the mapping between $\{(\monomialAmp[1],\mydots,\monomialAmp[\numberOfIterations])\}$ and $\{(\monomial[1],\mydots,\monomial[\numberOfIterations])\}$ is bijective and results in a Gray code when the elements of the set $\{(\monomial[1],\mydots,\monomial[\numberOfIterations])\}$ are ordered  lexicographically~\cite{sahin_2020gm}. Hence, the function $\funcfForFinalAmplitude(\seqx)$ is also similar to the first-order \ac{RM} code, except that the operations occur in $\realNumbers$. We refer the readers to \cite{parker_2003} for various representations of \acp{CS} and their properties.

\subsection{Signal Model and Wireless Channel}
We assume that the sensors and the \ac{UAV}  are equipped with a single antenna.  Let  $\transmittedSeq[\indexED]=(\transmittedSeqEle[\indexED,0],\mydots,\transmittedSeqEle[\indexED,\lengthOfSequence-1])$ be a  \ac{CS} of length $\lengthOfSequence$ transmitted from the $\indexED$th sensor over an \ac{OFDM} symbol by mapping its elements to a set of contiguous subcarriers. Assuming that all sensors access the wireless channel simultaneously and the \ac{CP} duration is larger than the sum of the maximum time-synchronization error and the maximum-excess delay of the channel, we can express the polynomial representation of the received sequence $\receivedSeq[]=(\receivedSeqEle[0],\mydots,\receivedSeqEle[\lengthOfSequence-1])$ at the \ac{UAV} after the signal superposition as
\begin{align}
\polySeq[{\receivedSeqP}][\polyVariable] &=\sum_{\varMonomial=0}^{\lengthGaGb-1} \underbrace{\left(\sum_{\indexED=1}^{\numberOfEdgeDevices}
\channelAtSubcarrier[\indexED,\varMonomial] \sqrt{\transmitPower[\indexED]}
\transmittedSeqEle[\indexED,\varMonomial]+\noiseAtSubcarrier[\varMonomial]\right)}_{\receivedSeqEle[\varMonomial]}
\polyVariable^{\varMonomial}~,	
	\label{eq:symbolOnSubcarrier}
\end{align}
where  $\channelAtSubcarrier[\indexED,\varMonomial]\sim\complexGaussian[0][1]$ is the Rayleigh fading channel coefficient between the \ac{UAV} and the $\indexED$th sensor for the $\varMonomial$th element of the sequence unless otherwise stated, $\transmitPower[\indexED]$ is the average transmit power, and ${\noiseAtSubcarrier[\varMonomial]}\sim\complexGaussian[0][\noiseVariance]$ is the \ac{AWGN}. 

We assume that the {\em average} received signal powers of the sensors at the \ac{UAV} are aligned with a power control mechanism. This assumption is weak as the impact of the large-scale channel model on the average received signal power can be tracked well with the state-of-the-art closed-loop power control loops by using control channels such as  \ac{PUCCH} or \ac{PRACH} in 3GPP \ac{5G} \ac{NR} \cite{10.5555/3294673}. As a result, the relative positions of the sensors to the UAV do not change our analyses. Note that a similar assumption is also made in \cite{Guangxu_2020, Guangxu_2021, Busra_2023}, where the time-variation in the channel is captured by the realizations of $\channelAtSubcarrier[\indexED,\varMonomial]$ in \eqref{eq:symbolOnSubcarrier}. In this study,  without loss of generality, we set $\transmitPower[\indexED]$, $\forall\indexED$, to $1$~Watt and calculate the \ac{SNR} of a sensor at the \ac{UAV} as $\SNR=1/\noiseVariance$.

\subsection{Problem Statement}

Suppose that the fading coefficient $\channelAtSubcarrier[\indexED,\varMonomial]$  is not available at the $\indexED$th sensor and the \ac{UAV} due to  synchronization impairments, reciprocity calibration errors, or  mobility. Under this constraint, the main objective of the \ac{UAV} is to compute the \ac{MV} $\majorityVoteEle[\indexCoordinate]$, $\forall\indexCoordinate$, by exploiting the signal superposition property of the multiple-access channels. Our main goal is to obtain a scheme that computes the MVs with a low probability of incorrect detection without using the fading coefficients at the sensors and the UAV, while the \ac{PMEPR} of the transmitted signal is guaranteed to be less than a certain value.
Although the \acp{CS} generated with Theorem~\ref{th:reduced} can address the latter challenge by keeping  the \ac{PMEPR} of transmitted OFDM signals at most $3$~dB, it is not trivial to use them for \ac{OAC}. 
In the following section, we show that Theorem 1 can be utilized to construct an OAC scheme to compute MVs without using the CSI at the sensors and the UAV.

\section{Proposed Scheme}
\label{sec:scheme}
%
Given the number of parameters that can be chosen independently in Theorem~\ref{th:reduced}, we consider $\numberOfIterations$ MV computations. Let $\voteVectorED[\indexED]$ be the vector of $\numberOfIterations$ votes of the $\indexED$th sensor, i.e., $(\voteVectorEDEle[\indexED][1],\mydots,\voteVectorEDEle[\indexED][\numberOfIterations])$ for $\voteVectorEDEle[\indexED][\indexIteration]\in\{-1,0,1\}$, $\forall\indexIteration$. If $\voteVectorEDEle[\indexED][\indexIteration]=0$, i.e., an absentee vote, the $\indexED$th sensor does not participate in the $\indexIteration$th \ac{MV} computation. Note that the absentee votes have previously been shown to be useful  for addressing data heterogeneity for wireless federated learning along with \ac{OAC}  \cite{Sahin_2022MVjournal}.  In particular, to address the scenario discussed in Section~\ref{sec:system},  we set $\voteVectorEDEle[\indexED][\indexIteration]=\voteVectorEDEleCoordinate[\indexED,\indexIteration][\indexRound]$ for $\indexIteration=\{1,2,3\}$ and $\voteVectorEDEle[\indexED][\indexIteration]=0$ for $\indexIteration\in\{4,5,\mydots,\numberOfIterations\}$ without loss of generality, unless otherwise stated.

The proposed scheme modulates the amplitude of the elements of the \ac{CS} via $\funcfForFinalAmplitude(\seqx)$ as a function of the votes $\voteVectorED[\indexED]$ at the $\indexED$th sensor. To this end,  based on Theorem~\ref{th:reduced}, let us denote the functions used at the $\indexED$th sensor as $\funcfForFinalAmplitudeED[\indexED](\seqx)$ and $\funcfForFinalPhaseED[\indexED](\seqx)$, and their parameters as $\{\scaleEexpAtRound[\indexED,0][\indexRound],\scaleEexpAtRound[\indexED,1][\indexRound],\mydots,\scaleEexpAtRound[\indexED,\numberOfIterations][\indexRound]\}$ and $\{\angleexpAllAtRound[\indexED,0][\indexRound],\angleexpAllAtRound[\indexED,1][\indexRound],\mydots,\angleexpAllAtRound[\indexED,\numberOfIterations][\indexRound]\}$, respectively. To synthesize the transmitted sequence $\transmittedSeq[\indexED]$ of length $\lengthGaGb=2^\numberOfIterations$, we use a fixed permutation $\seqPermutationCompShift$ and  map $\voteVectorEDEle[\indexED][\indexIteration]$ to $\scaleEexpAtRound[\indexED,\indexIteration][\indexRound]$ as 
\begin{align}
	\scaleEexpAtRound[\indexED,\indexIteration][\indexRound]=\begin{cases}
		-\scalingParameters,&\voteVectorEDEle[\indexED][\indexIteration]=-1\\
		0,&\voteVectorEDEle[\indexED][\indexIteration]=0\\
		+\scalingParameters,&\voteVectorEDEle[\indexED][\indexIteration]=+1
	\end{cases},~\forall\indexIteration~,
\label{eq:modulation}
\end{align}
where $\scalingParameters>0$ is a scaling parameter. To ensure that the squared $\ell_2$-norm of the \ac{CS} $\transmittedSeq[\indexED]$ is  $2^\numberOfIterations$, i.e., $\norm{\transmittedSeq[\indexED]}_2^2=2^\numberOfIterations$, we choose $\scaleEexpAtRound[\indexED,0][\indexRound]$ as
\begin{align}
\scaleEexpAtRound[\indexED,0][\indexRound] =- \frac{1}{2}\sum_{\indexIteration=1}^{\numberOfIterations}\ln{\frac{1+\constante^{2\scaleEexpAtRound[\indexED,\indexIteration][\indexRound]}}{2}}.
\label{eq:normalizationCoef}
\end{align}
To derive \eqref{eq:normalizationCoef}, notice that $\scaleEexpAtRound[\indexED,\indexIteration][\indexRound]$ scales $2^{\numberOfIterations-1}$ elements of the \ac{CS} by $\constante^{\scaleEexpAtRound[\indexED,\indexIteration][\indexRound]}$ in \eqref{eq:realPartReduced} for $\indexIteration>0$. Therefore, $\norm{\transmittedSeq[\indexED]}_2^2$ is scaled by ${(1+\constante^{2\scaleEexpAtRound[\indexED,\indexIteration][\indexRound]}})/{2}$. By considering $\scaleEexpAtRound[\indexED,1][\indexRound],\mydots,\scaleEexpAtRound[\indexED,\numberOfIterations][\indexRound]$, the total scaling can be calculated as $\totalPowerScale\triangleq\prod_{\indexIteration=1}^{\numberOfIterations}{(1+\constante^{2\scaleEexpAtRound[\indexED,\indexIteration][\indexRound]}})/{2}$. Thus, $\constante^{2\scaleEexpAtRound[\indexED,0][\indexRound]}=1/\totalPowerScale$ must hold for $\norm{\transmittedSeq[\indexED]}_2^2=2^\numberOfIterations$, which results in \eqref{eq:normalizationCoef}.

With \eqref{eq:modulation} and \eqref{eq:normalizationCoef}, if $\voteVectorEDEle[\indexED][\indexIteration]\neq0$ for $\scalingParameters\rightarrow\infty$, one half of elements (i.e., the ones for $
\monomialAmp[{\permutationMono[{\indexIteration}]}]=0$)  of the \ac{CS} $\transmittedSeq[\indexED]$  are set to $0$  while the other half (i.e., the ones for $
\monomialAmp[{\permutationMono[{\indexIteration}]}]=1$) are scaled by a factor of $\sqrt{2}$ and the sign of $\voteVectorEDEle[\indexED][\indexIteration]$ determines which half is amplified. For $\voteVectorEDEle[\indexED][\indexIteration]=0$, the halves are not scaled.
\begin{table}[t]
	\caption{An example of encoded CSs based on votes for $\numberOfIterations=3$.}
	\centering
	\resizebox{3.5in}{!}{	
		\begin{tabular}{l|c|c|c|c|c|c|c|c}
			$\voteVectorED[\indexED]$ & $\transmittedSeqEle[\indexED,0]$ & $\transmittedSeqEle[\indexED,1]$ & $\transmittedSeqEle[\indexED,2]$ & $\transmittedSeqEle[\indexED,3]$  & $\transmittedSeqEle[\indexED,4]$  & $\transmittedSeqEle[\indexED,5]$ & $\transmittedSeqEle[\indexED,6]$ & $\transmittedSeqEle[\indexED,7]$ \\\hline
			$(0, 0, 0)$	& $1$ & $1$ & $1$ & $-1$ & $1$ & $1$ & $-1$ & $1$	\\	
			$(1, 0, 0)$	& $0$ & $\sqrt{2}$ & $\sqrt{2}$ & $0$ & $0$ & $\sqrt{2}$ & $-\sqrt{2}$ & $0$  \\
			$(1, 1, 0)$	& $0$ & $0$ & $2$ & $0$ & $0$ & $2$ & $0$ & $0$  \\
			$(1, 1, 1)$	& $0$ & $0$ & $0$ & $0$ & $0$ & $2\sqrt{2}$ & $0$ & $0$  \\
			$(1, 1, -1)$& $0$ & $0$ & $2\sqrt{2}$ & $0$ & $0$ & $0$ & $0$ & $0$\\
			$(1, -1, 0)$ & $0$ & $2$ & $0$ & $0$ & $0$ & $0$ & $-2$ & $0$  \\
			$(-1, 0, 0)$ & $\sqrt{2}$ & $0$ & $0$ & $-\sqrt{2}$ & $\sqrt{2}$ & $0$ & $0$ & $\sqrt{2}$  \\	
		\end{tabular}
	}
	\label{table:example}
\end{table}
\begin{example}
	\label{ex:sequences} \rm
Let $\seqPermutationCompShift=(3,2,1)$, $\numberOfPointsForPSK=2$, $\numberOfIterations=3$, $\angleexpAllAtRound[\indexIteration][\indexRound]=0$, $\forall\indexIteration$. Hence, the indices of the scaled elements are controlled by $\monomialAmp[{1}]=\monomial[{1}]$, $\monomialAmp[{2}]=(\monomial[{1}] +\monomial[{2}])_2$, and $\monomialAmp[3]=(\monomial[2] +\monomial[3])_2$ when $(\monomial[1],\monomial[2],\monomial[3])$ is listed in lexicographic order, i.e., $(0,0,0),(0,0,1),\mydots,(1,1,1)$.  The encoded \acp{CS} for  several realizations of  $\voteVectorED[\indexED]$ for $\scalingParameters\rightarrow\infty$ are given in \tablename~\ref{table:example}. For $\voteVectorED[\indexED]=(1, 0, 0)$ and $\voteVectorED[\indexED]=(-1, 0, 0)$, four elements determined by $\monomialAmp[3]$ of the uni-modular \ac{CS}  is scaled by $\sqrt{2}$, and the rest is multiplied by $0$. 
Similarly, for $\voteVectorED[\indexED]=(1, 1, 1)$ and $\voteVectorED[\indexED]=(1, 1, -1)$, four elements of the \ac{CS} (i.e., the CS for $\voteVectorED[\indexED]=(1, 1, 0)$) is scaled by $\sqrt{2}$, and the rest is multiplied with $0$. It is worth noting that if all the votes are non-zero, only one of the eight elements of the sequence is non-zero. 
\end{example}

For the proposed scheme, the values for  $\angleexpAllAtRound[\indexED,0][\indexRound],\angleexpAllAtRound[\indexED,1][\indexRound],\mydots,\angleexpAllAtRound[\indexED,\numberOfIterations][\indexRound]$ are chosen randomly from the set $\integers_\numberOfPointsForPSK$ for the randomization of $\transmittedSeq[\indexED]$ across the sensors. This choice is also in line with the cases where phase synchronization cannot be maintained in the network.

Based on \eqref{eq:symbolOnSubcarrier}, the received sequence at the \ac{UAV} after signal superposition can be expressed as
\begin{align}
	\polySeq[{\receivedSeqP}][\polyVariable] &=\sum_{\forall\seqx\in\integers^\numberOfIterations_2} \underbrace{\left(\sum_{\indexED=1}^{\numberOfEdgeDevices}
		\channelAtSubcarrier[\indexED,\funcEnum(\seqx)]
		\constante^{\funcfForFinalAmplitudeED[\indexED](\seqx)}
		\constante^{\constantj\frac{2\pi}{\numberOfPointsForPSK}\funcfForFinalPhaseED[\indexED](\seqx)}+\noiseAtSubcarrier[\funcEnum(\seqx)]\right)}_{\receivedSeqEle[\funcEnum(\seqx)]}
	\polyVariable^{\funcEnum(\seqx)}~,	
	\label{eq:super}
\end{align}
The scaled halves of the transmitted sequences based on \eqref{eq:modulation} and \eqref{eq:normalizationCoef} non-coherently aggregate and the positions of the aggregated elements for the $\indexIteration$th MV are determined by $\monomialAmp[{\permutationMono[{\indexIteration}]}]$. Thus, to compute the $\indexIteration$th \ac{MV}, the \ac{UAV} calculates two metrics given by
\begin{align}
\metricPlus[\indexIteration] &\triangleq \sum_{\substack{\forall\seqx\in\integers^\numberOfIterations_2\\\monomialAmp[{\permutationMono[{\indexIteration}]}]=1}}|\receivedSeqEle[\funcEnum(\seqx)]|^2\nonumber\\&=\sum_{\substack{\forall\seqx\in\integers^\numberOfIterations_2\\\monomialAmp[{\permutationMono[{\indexIteration}]}]=1}} \left|\sum_{\indexED=1}^{\numberOfEdgeDevices}	\channelAtSubcarrier[\indexED,\funcEnum(\seqx)] 	\constante^{\funcfForFinalAmplitudeED[\indexED](\seqx)}	\constante^{\constantj\frac{2\pi}{\numberOfPointsForPSK}\funcfForFinalPhaseED[\indexED](\seqx)}+\noiseAtSubcarrier[\funcEnum(\seqx)]\right|^2,
\label{eq:metricPlus}
\end{align}
and
\begin{align}
	\metricMinus[\indexIteration] &\triangleq \sum_{\substack{\forall\seqx\in\integers^\numberOfIterations_2\\\monomialAmp[{\permutationMono[{\indexIteration}]}]=0}}|\receivedSeqEle[\funcEnum(\seqx)]|^2\nonumber\\&=\sum_{\substack{\forall\seqx\in\integers^\numberOfIterations_2\\\monomialAmp[{\permutationMono[{\indexIteration}]}]=0}} \left|\sum_{\indexED=1}^{\numberOfEdgeDevices}	\channelAtSubcarrier[\indexED,\funcEnum(\seqx)] 	\constante^{\funcfForFinalAmplitudeED[\indexED](\seqx)}	\constante^{\constantj\frac{2\pi}{\numberOfPointsForPSK}\funcfForFinalPhaseED[\indexED](\seqx)}+\noiseAtSubcarrier[\funcEnum(\seqx)]\right|^2.
	\label{eq:metricMinus}
\end{align}
It then detects the $\indexIteration$th \ac{MV} by comparing the values of $\metricPlus[\indexIteration]$ and $\metricMinus[\indexIteration]$ as
\begin{align}
	\majorityVoteDetectedEle[\indexIteration]=\signNormal[{\metricPlus[\indexIteration]-\metricMinus[\indexIteration]}],\forall\indexIteration~.
	\label{eq:decision}
\end{align}

In \figurename~\ref{fig:systemModel}, we provide the transmitter and receiver block diagrams for the proposed \ac{OAC} scheme. The $\indexED$th sensor first estimates the position of the \ac{UAV}, e.g., by using some image processing \cite{MUSTAFAH2012575,Carrio_2020}. It computes the vector $\voteVectorED[\indexED]$ and  calculates $\transmittedSeq[\indexED]$ based on Theorem~\ref{th:reduced} by using the mapping in \eqref{eq:modulation}. It then maps the elements of the encoded \ac{CS} $\transmittedSeq[\indexED]$ to $2^\numberOfIterations$ \ac{OFDM} subcarriers and calculates the $N$-point \ac{IDFT} of the mapped \ac{CS}. All the sensors transmit their  signals along with a sufficiently large \ac{CP} duration for \ac{OAC}. The \ac{UAV} receives the non-coherently superposed signal. After discarding the \ac{CP} and calculating the DFT of the remaining received samples, it obtains $\metricPlus[\indexIteration]$ and $\metricMinus[\indexIteration]$. The UAV finally detects the \acp{MV} with \eqref{eq:decision} and updates its position based on \eqref{eq:dynamics}. We discuss the detector performance in \eqref{eq:decision} rigorously in the following section.


\begin{figure}
	\centering
	\includegraphics[width =1.8in]{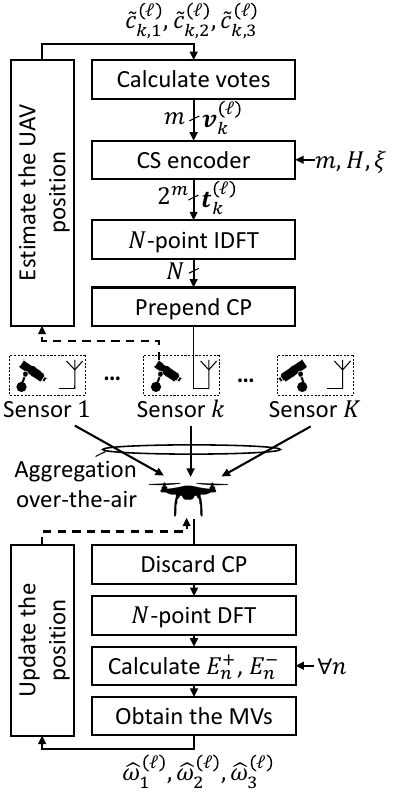}
	\caption{Transmitter and receiver diagrams for the proposed OAC scheme.}
	\label{fig:systemModel}
\end{figure}

\section{Performance Analysis}
\label{sec:performance}

\subsection{Average Performance}
Let $\numberOfEDsPlus[\indexIteration]$, $\numberOfEDsMinus[\indexIteration]$, and $\numberOfEDsZero[\indexIteration]$ be the number of sensors with positive, negative, and zero votes for $\indexIteration$th \ac{MV} computation, respectively. 
\begin{lemma}
$\expectationOperator[{\metricPlus[\indexIteration]}][]$ and $\expectationOperator[{\metricMinus[\indexIteration]}][]$ can be calculated as
\begin{align}
	\expectationOperator[{\metricPlus[\indexIteration]}][]
&=\frac{2^\numberOfIterations\constante^{2\scalingParameters}\numberOfEDsPlus[\indexIteration]}{1+\constante^{2\scalingParameters}}+\frac{2^\numberOfIterations\constante^{-2\scalingParameters}\numberOfEDsMinus[\indexIteration]}{1+\constante^{-2\scalingParameters}}+{2^{\numberOfIterations-1}}(\numberOfEDsZero[\indexIteration]+\noiseVariance)	\nonumber,
\\
\expectationOperator[{\metricMinus[\indexIteration]}][]
&=\frac{2^\numberOfIterations\numberOfEDsPlus[\indexIteration]}{1+\constante^{2\scalingParameters}}+\frac{2^\numberOfIterations\numberOfEDsMinus[\indexIteration]}{1+\constante^{-2\scalingParameters}}+{2^{\numberOfIterations-1}}(\numberOfEDsZero[\indexIteration]+\noiseVariance)	\nonumber,
\end{align}
respectively, where the expectation is over the distribution of channel and noise.
\label{lemma:exp}
\end{lemma}
The proof is given in Appendix~\ref{prof:lemma:exp}. 

Without any concern about the norm of  $\transmittedSeq[\indexED]$ with \eqref{eq:normalizationCoef}, we can  choose an arbitrarily large  $\scalingParameters$, leading to the following result:
\begin{corollary}
The following identities hold:
\begin{align}
\lim_{\scalingParameters\rightarrow\infty}\expectationOperator[{\metricPlus[\indexIteration]}][]&=2^\numberOfIterations\numberOfEDsPlus[\indexIteration]+2^{\numberOfIterations-1}\numberOfEDsZero[\indexIteration]+2^{\numberOfIterations-1}\noiseVariance~,\nonumber\\
\lim_{\scalingParameters\rightarrow\infty}\expectationOperator[{\metricMinus[\indexIteration]}][]&=2^\numberOfIterations\numberOfEDsMinus[\indexIteration]+2^{\numberOfIterations-1}\numberOfEDsZero[\indexIteration]+2^{\numberOfIterations-1}\noiseVariance~. \nonumber
\end{align}
\label{corol:exp}
\end{corollary} 
Based on Corollary~\ref{corol:exp}, we can infer that the detector in \eqref{eq:decision} is likely to detect the correct \ac{MV} for $\scalingParameters\rightarrow\infty$ since $\lim_{\scalingParameters\rightarrow\infty}\expectationOperator[{{\metricPlus[\indexIteration]-\metricMinus[\indexIteration]}}][]=2^{\numberOfIterations-1}(\numberOfEDsPlus[\indexIteration]-\numberOfEDsMinus[\indexIteration])$ holds. Also, the impact of absentee votes on the metrics $\metricPlus[\indexIteration]$ and $\metricMinus[\indexIteration]$ are equally shared.

\subsection{Computation Error Rate}
For a given set of all votes $\voteAll\triangleq(\voteVectorAcrossED[1],\mydots,\voteVectorAcrossED[\numberOfIterations])$  for $\voteVectorAcrossED[\indexIteration]\triangleq(\voteVectorEDEle[1][\indexIteration],\mydots,\voteVectorEDEle[\numberOfEdgeDevices][\indexIteration])$, the \ac{CER}  can be defined as
\begin{align}
	\computationErrorRate[\voteAll] \triangleq & \begin{cases}
		\probability[{\metricPlus[\indexIteration]-\metricMinus[\indexIteration]<0;\voteAll}]~, & \numberOfEDsPlus[\indexIteration]>\numberOfEDsMinus[\indexIteration]\\
		\probability[{\metricPlus[\indexIteration]-\metricMinus[\indexIteration]>0;\voteAll}]~, &
		\numberOfEDsPlus[\indexIteration]<\numberOfEDsMinus[\indexIteration]\\
		1~, & \numberOfEDsPlus[\indexIteration]=\numberOfEDsMinus[\indexIteration]
	\end{cases}~,\nonumber
	\\
	=&
	\begin{cases}
		\CDF[{\metricPlus[\indexIteration]-\metricMinus[\indexIteration]}][][0;\voteAll]~, & \numberOfEDsPlus[\indexIteration]>\numberOfEDsMinus[\indexIteration]\\
		1-\CDF[{\metricPlus[\indexIteration]-\metricMinus[\indexIteration]}][][0;\voteAll]~, &
		\numberOfEDsPlus[\indexIteration]<\numberOfEDsMinus[\indexIteration]\\
		1~, & \numberOfEDsPlus[\indexIteration]=\numberOfEDsMinus[\indexIteration]
	\end{cases}~,
	\label{eq:cerMain}
\end{align}
where $\CDF[{\metricPlus[\indexIteration]-\metricMinus[\indexIteration]}][][x;\voteAll]$ is the \ac{CDF} of $\metricPlus[\indexIteration]-\metricMinus[\indexIteration]$. It is worth noting that the detector in \eqref{eq:decision} always makes an error due to the noisy reception in communication channels when $\numberOfEDsPlus[\indexIteration]$ and $\numberOfEDsMinus[\indexIteration]$ are identical, leading to the third case in \eqref{eq:cerMain}. For a given $\voteAll$, the \ac{CDF} of $\metricPlus[\indexIteration]-\metricMinus[\indexIteration]$ can be obtained as follows:
\begin{lemma}
	Suppose  $\channelAtSubcarrier[\indexED,\varMonomial]\sim\complexGaussian[0][1]$ (i.e., frequency-selective fading) holds. $\CDF[{\metricPlus[\indexIteration]-\metricMinus[\indexIteration]}][][x;\voteAll]$ can be calculated as
	\begin{align}
		\CDF[{\metricPlus[\indexIteration]-\metricMinus[\indexIteration]}][][x;\voteAll]&=	\frac{1}{2}-\int_{-\infty}^{\infty}\frac{\charFcnPlus[\integralVar]\charFcnMinusConj[\integralVar]}{2\pi\constantj\integralVar} \constante^{-\constantj\integralVar\CDFvariable} d\integralVar~,
		\label{eq:cdfMain}
	\end{align}
	respectively, 	where $\charFcnPlus[\integralVar]$ and $\charFcnMinus[\integralVar]$ are given by
	\begin{align}
		\charFcnPlus[\integralVar]\triangleq\prod_{{\substack{\forall\seqx\in\integers^\numberOfIterations_2\\\monomialAmp[{\permutationMono[{\indexIteration}]}]=1}}}\frac{1}{1-\constantj\integralVar\rate[{\funcEnum(\seqx)}]^{-1}}~, 
	\end{align}
	and
	\begin{align}
		\charFcnMinus[\integralVar]\triangleq\prod_{{\substack{\forall\seqx\in\integers^\numberOfIterations_2\\\monomialAmp[{\permutationMono[{\indexIteration}]}]=0}}}\frac{1}{1-\constantj\integralVar\rate[{\funcEnum(\seqx)}]^{-1}}, 
	\end{align}
	respectively, for $
	\rate[{\funcEnum(\seqx)}]^{-1} \triangleq \sum_{\indexED=1}^{\numberOfEdgeDevices}
	\constante^{2\funcfForFinalAmplitudeED[\indexED](\seqx)}+\noiseVariance
	$.
	\label{lemma:errProbGivenVote}
\end{lemma}

The proof is given in Appendix~\ref{prof:lemma:errProbGivenVote}.

Although $\CDF[{\metricPlus[\indexIteration]-\metricMinus[\indexIteration]}][][x;\voteAll]$ in \eqref{eq:cdfMain} is not a closed-form expression, it can be easily evaluated with a numeric integration to compute $\computationErrorRate[\voteAll]$ in \eqref{eq:cerMain}. Also, as demonstrated in Section~\ref{sec:numerical} (i.e., \figurename~\ref{fig:cerChannel}), \eqref{eq:cdfMain}  holds approximately in the flat-fading  fading channels since the values for  $\angleexpAllAtRound[\indexED,0][\indexRound],\angleexpAllAtRound[\indexED,1][\indexRound],\mydots,\angleexpAllAtRound[\indexED,\numberOfIterations][\indexRound]$ are chosen randomly in our approach.

Let $\probabilityPlusDecision$, $\probabilityMinusDecision$, and $\probabilityNullDecision$ denote the probabilities given by $ \probability[{\voteVectorEDEle[\indexED][\indexIteration]>0}]$, $ \probability[{\voteVectorEDEle[\indexED][\indexIteration]<0}]$, and $ \probability[{\voteVectorEDEle[\indexED][\indexIteration]= 0}]$, respectively,  $\forall\indexED\in\{1,\mydots,\numberOfEdgeDevices\}$ and  $\forall\indexIteration\in\{1,\mydots,\numberOfIterations\}$. 
\begin{corollary}
	\label{cor:cerpqzKpKnKz}
	For given $\numberOfEDsPlus[\indexIteration]$, $\numberOfEDsMinus[\indexIteration]$, $\numberOfEDsZero[\indexIteration]$,   $\probabilityPlusDecision$, $\probabilityMinusDecision$, and $\probabilityNullDecision$, the \ac{CER} for the $\indexIteration$th \ac{MV} is given by
	\begin{align}
		&\computationErrorRate[\indexIteration;{\numberOfEDsPlus[\indexIteration], \numberOfEDsMinus[\indexIteration], \numberOfEDsZero[\indexIteration]}] \nonumber\\&~~~~~~= \begin{cases}
		\probabilityErrorPlus[\indexIteration;{\numberOfEDsPlus[\indexIteration], \numberOfEDsMinus[\indexIteration], \numberOfEDsZero[\indexIteration]}]~, & \numberOfEDsPlus[\indexIteration]>\numberOfEDsMinus[\indexIteration]\\
		\probabilityErrorMinus[\indexIteration;{\numberOfEDsPlus[\indexIteration], \numberOfEDsMinus[\indexIteration], \numberOfEDsZero[\indexIteration]}]~, &
		\numberOfEDsPlus[\indexIteration]<\numberOfEDsMinus[\indexIteration]\\
		1~, & \numberOfEDsPlus[\indexIteration]=\numberOfEDsMinus[\indexIteration]
		\label{eq:cerKpKnKz}
	\end{cases}~,
	\end{align}
	where
	\begin{align}
		&\probabilityErrorPlus[\indexIteration;{\numberOfEDsPlus[\indexIteration], \numberOfEDsMinus[\indexIteration], \numberOfEDsZero[\indexIteration]}]
		\nonumber\\&~~\triangleq\frac{1}{2}-
		\sum_{\forall\voteAllWithout}
		\probabilityPlusDecision^\numberOfCorrectDecision\probabilityMinusDecision^\numberOfIncorrectDecision\probabilityNullDecision^\numberOfZeroDecision\int_{-\infty}^{\infty}\frac{\charFcnPlus[\integralVar]\charFcnMinusConj[\integralVar]}{2\pi\constantj\integralVar} d\integralVar~,\label{eq:cerAveragePlus}\\
		&\probabilityErrorMinus[\indexIteration;{\numberOfEDsPlus[\indexIteration], \numberOfEDsMinus[\indexIteration], \numberOfEDsZero[\indexIteration]}]
		\nonumber\\&~~\triangleq\frac{1}{2}+
		\sum_{\forall\voteAllWithout}
		\probabilityPlusDecision^\numberOfCorrectDecision\probabilityMinusDecision^\numberOfIncorrectDecision\probabilityNullDecision^\numberOfZeroDecision\int_{-\infty}^{\infty}\frac{\charFcnPlus[\integralVar]\charFcnMinusConj[\integralVar]}{2\pi\constantj\integralVar} d\integralVar~,\label{eq:cerAverageMinus} 
	\end{align}
	respectively, where $\numberOfCorrectDecision$, $\numberOfIncorrectDecision$, and $\numberOfZeroDecision$ are the number of elements  in $\voteAllWithout\triangleq(\voteVectorAcrossED[1],\mydots,\voteVectorAcrossED[\indexIteration-1],\voteVectorAcrossED[\indexIteration+1],\mydots,\voteVectorAcrossED[\numberOfIterations])$ with $1$, $-1$, and $0$, respectively.
\end{corollary}
\begin{proof}
	The probability of a realization of $\voteAllWithout$ is $\probabilityPlusDecision^\numberOfCorrectDecision\probabilityMinusDecision^\numberOfIncorrectDecision\probabilityNullDecision^\numberOfZeroDecision$. Thus, the expectation of \eqref{eq:cerMain} over the distribution of $\voteAllWithout$ leads to \eqref{eq:cerAveragePlus} and \eqref{eq:cerAverageMinus}.
\end{proof}
The calculations of \eqref{eq:cerAveragePlus} and \eqref{eq:cerAverageMinus} can be intractable due to the enumerations of $\voteAllWithout$. To address this issue, we average the \ac{CER} in \eqref{eq:cerMain} over a few realizations of $\voteAll$ for a given triplet $\{\probabilityPlusDecision,\probabilityMinusDecision,\probabilityNullDecision\}$ to compute \eqref{eq:cerAveragePlus} and \eqref{eq:cerAverageMinus} in Section~\ref{sec:numerical}.

Finally, let us define $\probabilityCER$ by setting it as 
\begin{align}
\probabilityCER= \probability[{\majorityVoteDetectedEle[\indexIteration]\neq\majorityVoteEle[\indexIteration]}]~,
\label{eq:CERdef}
\end{align}
for given $\probabilityPlusDecision$, $\probabilityMinusDecision$, and $\probabilityNullDecision$. We can calculate $\probabilityCER$ as follows:
\begin{corollary}
	\label{cor:cerpqz}
	For given $\probabilityPlusDecision$, $\probabilityMinusDecision$, and $\probabilityNullDecision$, $\probabilityCER$  is given by
	\begin{align}
		\probabilityCER &= \probability[{\majorityVoteDetectedEle[\indexIteration]=-1,\majorityVoteEle[\indexIteration]=1}]+\probability[{\majorityVoteDetectedEle[\indexIteration]=1,\majorityVoteEle[\indexIteration]=-1}]\nonumber\\&~~~+\probability[{\majorityVoteEle[\indexIteration]=0}]~,
		\label{eq:CERpqz}
	\end{align}
	where
	\begin{align}
		&\probability[{\majorityVoteDetectedEle[\indexIteration]=-1,\majorityVoteEle[\indexIteration]=1}]\nonumber\\
		&=\sum_{\varN=0}^{\numberOfEdgeDevices}	\sum_{\varP=\varN+1}^{\numberOfEdgeDevices-\varN}	\binom{\numberOfEdgeDevices}{\varP}	\binom{\numberOfEdgeDevices-\varP}{\varN}	\probabilityPlusDecision^{\varP}\probabilityMinusDecision^{\varN}\probabilityNullDecision^{\numberOfEdgeDevices-\varP-\varN} \probabilityErrorPlus[\indexIteration;{\varP, \varN, \varZ}]~,\nonumber
		\\
		&\probability[{\majorityVoteDetectedEle[\indexIteration]=1,\majorityVoteEle[\indexIteration]=-1}]\nonumber\\
		&=\sum_{\varP=0}^{\numberOfEdgeDevices}\sum_{\varN=\varP+1}^{\numberOfEdgeDevices-\varP}\binom{\numberOfEdgeDevices}{\varP}\binom{\numberOfEdgeDevices-\varP}{\varN}\probabilityPlusDecision^{\varP}\probabilityMinusDecision^{\varN}\probabilityNullDecision^{\numberOfEdgeDevices-\varP-\varN} \probabilityErrorMinus[\indexIteration;{\varP, \varN, \varZ}]~,\nonumber
		\\
		&\probability[{\majorityVoteEle[\indexIteration]=0}]=\sum_{\varP=0}^{\floor{\frac{\numberOfEdgeDevices}{2}}}	\binom{\numberOfEdgeDevices}{\varP}	\binom{\numberOfEdgeDevices-\varP}{\varP}\probabilityPlusDecision^{\varP}\probabilityMinusDecision^{\varP}\probabilityNullDecision^{\numberOfEdgeDevices-2\varP}~.\nonumber
	\end{align}
\end{corollary}
The third term in \eqref{eq:CERpqz} is because of the third case in \eqref{eq:cerKpKnKz}. In Section~\ref{sec:numerical}, we evaluate $\probabilityCER$ for different $\probabilityPlusDecision$, $\probabilityMinusDecision$, and $\probabilityNullDecision$ values.

\subsection{Computation Rate and Resource Utilization Ratio}
The computation rate $\computationRate$ can be defined as the number of functions  computed per channel use (in real dimension) \cite{goldenbaum2015nomographic,Jeon_2016}. Since the proposed scheme computes $\numberOfIterations$ \acp{MV} over $2^{\numberOfIterations}$ complex-valued resources,  $\computationRate$ can be expressed as 
\begin{align}
	\computationRate =\frac{\numberOfIterations}{2^{\numberOfIterations+1}}~.
\end{align}
Hence, for a larger $\numberOfIterations$, the computation rate reduces while the \ac{CER} improves significantly as demonstrated in Section~\ref{sec:numerical}.

Let us define the resource utilization ratio as the ratio between the number of resources consumed with the proposed scheme and the number of resources when the communication and computation are considered as separate tasks, i.e., the traditional first-communicate-then-compute approach. For the separation, we assume that the spectral efficiency is $\rateTraditional$~bit/s/Hz. Hence, the  resources  needed for $\numberOfEdgeDevices$ EDs and $\numberOfIterations$ bits (i.e., votes) can be obtained as $\numberOfIterations\numberOfEdgeDevices/\rateTraditional$. The proposed scheme consumes $2^{\numberOfIterations}$ complex-valued resources for $\numberOfEdgeDevices$ EDs. Hence, the resource utilization ratio can be expressed as
\begin{align}
	\efficientImprovement = \frac{2^\numberOfIterations}{\numberOfIterations\numberOfEdgeDevices/\rateTraditional}~.
\end{align}
For instance, for $r=1$~bit/s/Hz, $\numberOfIterations=7$, and $\numberOfEdgeDevices=50$,  the wireless resources needed with the proposed scheme is $\efficientImprovement=0.3657$ times the ones  with the traditional approach.

\subsection{Complexity Analysis}
As can be seen from Theorem~\ref{th:reduced}, the functions $\funcfForFinalAmplitudeED[\indexED](\seqx)$ and $\funcfForFinalPhaseED[\indexED](\seqx)$ alter the coefficients of the first-order monomials. Hence, {\em without} a recursive method, $\numberOfIterations2^\numberOfIterations$ real-valued sum operations must be performed to calculate the functions $\funcfForFinalAmplitudeED[\indexED](\seqx)$ and $\funcfForFinalPhaseED[\indexED](\seqx)$,  and the multiplication operations can be replaced with multiplexers as $\seqx$ has binary elements. Also, $2^\numberOfIterations$ multiplications are required to compute the elements of $\transmittedSeq[\indexED]$. Thus, the computation complexity at the transmitter is $\mathcal{O}(2^\numberOfIterations)$. However, if the recursive construction of the \ac{RM} codes is exploited, the number of operations asymptotically scales with $\log_2(2^\numberOfIterations)$, resulting in a computation complexity of $\mathcal{O}(\numberOfIterations)$. At the receiver, 
\eqref{eq:metricPlus} and \eqref{eq:metricMinus} are calculated to obtain $\numberOfIterations$ \acp{MV}, which require at least $2^\numberOfIterations$ complex multiplications to compute the magnitude squares of the elements of the received sequence. Hence, the computation complexity at the receiver is $\mathcal{O}(2^\numberOfIterations)$, i.e., a linearly increased complexity with the length of \ac{CS}.

\section{Convergence Analysis}
\label{sec:convergence}
In this section, we discuss the convergence of the resulting systems under the control strategies in~\eqref{eq:strategies} by analyzing their Lyapunov stability based on the following definition:
\begin{definition}[\cite{lv2023sliding}]
	\rm
	A stochastic system $x^{(\ell+1)}=f\left(x^{(\ell)}\right)$, where $f$ describes the dynamics and $x^{(\ell)}$ is the state, is called globally uniformly ultimately bounded in mean square with ultimate bound $b$, if there exist positive $b$ and $L$ such that for any $\mathbb E\left[\|x^{(0)}\|^2\right] < \infty$, the inequality $\mathbb E\left[\|x^{(\ell)}\|^2\right] \le b$ holds for all $\ell\ge L$.
\end{definition}
We define $\distance[\indexCoordinate][\indexRound]$ by setting $\distance[\indexCoordinate][\indexRound]=\locationEle[\indexRound][\indexCoordinate] - \locationTargetEle[\indexCoordinate]$ and re-express \eqref{eq:dynamics}  as
\begin{equation} \label{eq:dynamics1}
	\distance[\indexCoordinate][\indexRound+1] = \distance[\indexCoordinate][\indexRound] -\Trefresh\velocityVectorEle[\indexRound][\indexCoordinate].
\end{equation}
for the convergence analysis of each case in \eqref{eq:strategies}.

\subsection{Case 1: Continuous Update  \& Ideal Communications}
The control strategy can be written~as
\begin{align*}
	\feedbackEle[\indexRound][\indexCoordinate] & = \frac{1}{\numberOfEdgeDevices} \sum_{\indexED=1}^\numberOfEdgeDevices \left( \locationEle[\indexRound][\indexCoordinate] + \locationEstimationErrorEle[\indexRound][\indexED,\indexCoordinate] - \locationTargetEle[\indexCoordinate]\right) = \distance[\indexCoordinate][\indexRound] + \frac{1}{\numberOfEdgeDevices} \sum_{\indexED=1}^\numberOfEdgeDevices \locationEstimationErrorEle[\indexRound][\indexED,\indexCoordinate]~.
\end{align*}
Since $\locationEstimationErrorEle[\indexRound][\indexED,\indexCoordinate]$ is Gaussian and the summands are independent of each other, $\frac{1}{\numberOfEdgeDevices} \sum_{\indexED=1}^\numberOfEdgeDevices \locationEstimationErrorEle[\indexRound][\indexED,\indexCoordinate]$ is also Gaussian.  Based on stochastic control theory~\cite{aastrom2012introduction}, as long as $|1-\updateRate\Trefresh| < 1$, the resulting closed-loop system is stable in a stochastic sense.

\def\energyFunction[#1]{V\left(#1\right)}
\subsection{Case 2: MV-based Update \& Ideal Communications}\label{subsec:idealMV}
 In this case, $|\feedbackEle[\indexRound][\indexCoordinate]|$ is at most 1. Therefore, as long as $\updateRate \le \maximumVelocity$, input saturation will not happen and the dynamic can be written as  
\begin{align} \label{eq:IMV}
	\distance[\indexCoordinate][\indexRound+1] & = \distance[\indexCoordinate][\indexRound] -\updateRate \Trefresh \: {\rm \sign}\left(\sum_{\indexED=1}^\numberOfEdgeDevices \voteVectorEDEleCoordinate[\indexED,\indexCoordinate][\indexRound]\right) 
	\\ & = \distance[\indexCoordinate][\indexRound] -\updateRate \Trefresh \:\feedbackEle[\indexRound][\indexCoordinate]
	. \nonumber
\end{align}
Recall that $\voteVectorEDEleCoordinate[\indexED,\indexCoordinate][\indexRound]=\signNormal[{\locationEstimateEle[\indexRound][\indexED,\indexCoordinate]-\locationTargetEle[\indexCoordinate]}] = {\rm sign}\left(\distance[\indexCoordinate][\indexRound] + \locationEstimationErrorEle[\indexRound][\indexED,\indexCoordinate] \right)$ and $\locationEstimationErrorEle[\indexRound][\indexED,\indexCoordinate]$ is Gaussian. We can then obtain the values of $\probabilityPlusDecision$, $\probabilityMinusDecision$, and $\probabilityNullDecision$ as
\begin{align}
	\probabilityPlusDecision&={\rm Pr}(\voteVectorEDEleCoordinate[\indexED,\indexCoordinate][\indexRound] = +1) =  {\rm Pr}(\distance[\indexCoordinate][\indexRound] + \locationEstimationErrorEle[\indexRound][\indexED,\indexCoordinate] > 0) \nonumber \\ & = {\rm Pr}\left(\frac{\locationEstimationErrorEle[\indexRound][\indexED,\indexCoordinate]}{\stdSensor}  >  \frac{-\distance[\indexCoordinate][\indexRound]}{\stdSensor} \right) = 1 - \normalCDF[\frac{-\distance[\indexCoordinate][\indexRound]}{\stdSensor}]~, 
	\label{eq:alphaExp}\\
	\probabilityMinusDecision&={\rm Pr}(\voteVectorEDEleCoordinate[\indexED,\indexCoordinate][\indexRound] = -1)  =  {\rm Pr}(\distance[\indexCoordinate][\indexRound] + \locationEstimationErrorEle[\indexRound][\indexED,\indexCoordinate] < 0)\nonumber \\& = {\rm Pr}\left(\frac{\locationEstimationErrorEle[\indexRound][\indexED,\indexCoordinate]}{\stdSensor}  <  \frac{-\distance[\indexCoordinate][\indexRound]}{\stdSensor} \right)  = \normalCDF[\frac{-\distance[\indexCoordinate][\indexRound]}{\stdSensor}] ~,\label{eq:betaExp}\\
	\probabilityNullDecision&={\rm Pr}(\voteVectorEDEleCoordinate[\indexED,\indexCoordinate][\indexRound] = 0)  =  0~,\nonumber
\end{align}
respectively. With the distribution of $\voteVectorEDEleCoordinate[\indexED,\indexCoordinate][\indexRound]$, we have
\begin{align}
	\phiFunc[{+}][{\distance[\indexCoordinate][\indexRound]}] &\triangleq {\rm Pr}\left(\feedbackEle[\indexRound][\indexCoordinate] = +1\right) = \sum_{\anIndexForK=0}^{\ceil{\frac{\numberOfEdgeDevices}{2}}-1} \binom{\numberOfEdgeDevices}{\anIndexForK} \probabilityPlusDecision^{\numberOfEdgeDevices-\anIndexForK}\probabilityMinusDecision^{\anIndexForK}~, \label{eq:plus}
	\\
	\phiFunc[{-}][{\distance[\indexCoordinate][\indexRound]}] &\triangleq {\rm Pr}\left(\feedbackEle[\indexRound][\indexCoordinate]  =-1\right) = \sum_{\anIndexForK=0}^{\ceil{\frac{\numberOfEdgeDevices}{2}}-1} \binom{\numberOfEdgeDevices}{\anIndexForK} \probabilityPlusDecision^{\anIndexForK}\probabilityMinusDecision^{\numberOfEdgeDevices-\anIndexForK}~,	 \label{eq:minus}
		\\
\phiFunc[0][{\distance[\indexCoordinate][\indexRound]}]& \triangleq 
\begin{cases}
	0, & {\rm odd~}K \\ 
	\binom{\numberOfEdgeDevices}{{\numberOfEdgeDevices}/{2}}
	\probabilityPlusDecision^{\frac{\numberOfEdgeDevices}{2}} \probabilityMinusDecision^{\frac{\numberOfEdgeDevices}{2}}, & {\rm even~}\numberOfEdgeDevices
\end{cases}.	
\end{align}
Given the distribution of $\feedbackEle[\indexRound][\indexCoordinate]$, we can show the convergence of the system under the MV (ideal) case:
\begin{theorem} \label{thm2}
	Given $\updateRate \le \maximumVelocity$, the system in~\eqref{eq:dynamics1} is globally uniformly ultimately bounded in mean square under the MV (ideal) control strategy in~\eqref{eq:strategies}.
\end{theorem}
\begin{proof}
	Let $\energyFunction[x] \triangleq x^2$.  Then, by using \eqref{eq:IMV},
	\begin{align}
		&  \expectationOperator[{\energyFunction[{\distance[\indexCoordinate][\indexRound+1]}]|\distance[\indexCoordinate][\indexRound]}][] - \energyFunction[{\distance[\indexCoordinate][\indexRound]}] \label{eq:thm2-1}
		\\
		&= -2 \updateRate \Trefresh   \expectationOperator[{\feedbackEle[\indexRound][\indexCoordinate]|\distance[\indexCoordinate][\indexRound]}][] \distance[\indexCoordinate][\indexRound] + \updateRate^2 \Trefresh^2 \expectationOperator[{(\feedbackEle[\indexRound][\indexCoordinate])^2|\distance[\indexCoordinate][\indexRound]}][]. \nonumber
	\end{align}
By using \eqref{eq:plus} and \eqref{eq:minus}, 
\begin{align}
	&\expectationOperator[{\feedbackEle[\indexRound][\indexCoordinate]|\distance[\indexCoordinate][\indexRound]}][] = \phiFunc[+][{\distance[\indexCoordinate][\indexRound]}] - \phiFunc[-][{\distance[\indexCoordinate][\indexRound]}]~, \label{eq:exp1}\\
	& \expectationOperator[{(\feedbackEle[\indexRound][\indexCoordinate])^2|~\distance[\indexCoordinate][\indexRound]}][] = \phiFunc[+][{\distance[\indexCoordinate][\indexRound]}] + \phiFunc[-][{\distance[\indexCoordinate][\indexRound]}] \in [0,1]~.\label{eq:exp2}
\end{align}
Therefore, by using \eqref{eq:exp1} and \eqref{eq:exp2}, \eqref{eq:thm2-1} yields
\begin{align*}
	&  \expectationOperator[{\energyFunction[{\distance[\indexCoordinate][\indexRound+1]}]|\distance[\indexCoordinate][\indexRound]}][]  - \energyFunction[{\distance[\indexCoordinate][\indexRound]}] \\
	&= -2 \updateRate \Trefresh \left(\phiFunc[+][{\distance[\indexCoordinate][\indexRound]}] - \phiFunc[-][{\distance[\indexCoordinate][\indexRound]}]\right) \distance[\indexCoordinate][\indexRound] \\&~~~~ + \updateRate^2\Trefresh^2 \left( \phiFunc[+][{\distance[\indexCoordinate][\indexRound]}] + \phiFunc[-][{\distance[\indexCoordinate][\indexRound]}] \right)\\
	&\le -2 \updateRate\Trefresh \left(\phiFunc[+][{\distance[\indexCoordinate][\indexRound]}] - \phiFunc[-][{\distance[\indexCoordinate][\indexRound]}]\right) \distance[\indexCoordinate][\indexRound] + \updateRate^2 \Trefresh^2~.
\end{align*}
Note that $\distance[\indexCoordinate][\indexRound] > 0 \Rightarrow \phiFunc[+][{\distance[\indexCoordinate][\indexRound]}] > \phiFunc[-][{\distance[\indexCoordinate][\indexRound]}]$ and $\distance[\indexCoordinate][\indexRound] < 0 \Rightarrow \phiFunc[+][{\distance[\indexCoordinate][\indexRound]}] < \phiFunc[-][{\distance[\indexCoordinate][\indexRound]}]$, which means
\[
\left(\phiFunc[+][{\distance[\indexCoordinate][\indexRound]}] - \phiFunc[-][{\distance[\indexCoordinate][\indexRound]}]\right) \distance[\indexCoordinate][\indexRound] > 0~,
\]
always holds. Meanwhile, $\left(\phiFunc[{+}][{\distance[\indexCoordinate][\indexRound]}] - \phiFunc[{-}][{\distance[\indexCoordinate][\indexRound]}]\right) \distance[\indexCoordinate][\indexRound]$ is monotonic increasing when $\distance[\indexCoordinate][\indexRound] >0$ and monotonic decreasing when $\distance[\indexCoordinate][\indexRound] < 0$  with the global minimum at $\distance[\indexCoordinate][\indexRound] =0$. Thus, we can define
	\[
	\gamma = \max_{q \in \realNumbers}|q|,~~{\rm s.t.}~~ \left(\phiFunc[{+}][q] - \phiFunc[{-}][q]\right) q = \frac{\updateRate\Trefresh}{2}~,
	\]
	and the parameter $\gamma$ must be finite, whose value depends on $\updateRate\Trefresh$ and $\stdSensor$.
	For any $\distance[\indexCoordinate][\indexRound]$ such that $|\distance[\indexCoordinate][\indexRound]| > \gamma$, we have 
	\[
	\left(\phiFunc[{+}][{\distance[\indexCoordinate][\indexRound]}] - \phiFunc[{-}][{\distance[\indexCoordinate][\indexRound]}]\right) \distance[\indexCoordinate][\indexRound] > \frac{\updateRate\Trefresh}{2}~,
	\]
	and then
	\[
	 \expectationOperator[{\energyFunction[{\distance[\indexCoordinate][\indexRound+1]}]|\distance[\indexCoordinate][\indexRound]}][]  - \energyFunction[{\distance[\indexCoordinate][\indexRound]}] < 0~.
	\]
	Thus, the system is mean-square globally uniformly ultimately bounded and the mean-square ultimate bound is determined by $\gamma$.
\end{proof}

\subsection{Case 3: MV with the Proposed OAC}
By using Corollary~\ref{cor:cerpqzKpKnKz}, we can re-calculate $\phiFunc[{+}][{\distance[\indexCoordinate][\indexRound]}]$ and $\phiFunc[{-}][{\distance[\indexCoordinate][\indexRound]}]$ defined in \eqref{eq:plus} and \eqref{eq:minus}, respectively, as
\begin{align*}
	&\phiFunc[{+}][{\distance[\indexCoordinate][\indexRound]}] =1-
	\underbrace{\sum_{\anIndexForK=0}^{\numberOfEdgeDevices}	\binom{\numberOfEdgeDevices}{\anIndexForK}		\probabilityPlusDecision^{\anIndexForK}\probabilityMinusDecision^{\numberOfEdgeDevices-\anIndexForK} \probabilityErrorPlus[\indexIteration;{\anIndexForK, \numberOfEdgeDevices-\anIndexForK, 0}]}_{\triangleq\positiveHalfOAC}~,
	\\
	&\phiFunc[{-}][{\distance[\indexCoordinate][\indexRound]}] =1-	\phiFunc[{+}][{\distance[\indexCoordinate][\indexRound]}]~,
\end{align*}
where $\probabilityPlusDecision$ and $\probabilityMinusDecision$ are given in \eqref{eq:alphaExp} and \eqref{eq:betaExp}, respectively.
Now the convergence of the system under the OAC (MV) strategy is presented in the following theorem:
\begin{theorem}
	Given $\updateRate \le \maximumVelocity$, the system in~\eqref{eq:dynamics1} is globally uniformly ultimately bounded in mean square under the MV (OAC) control strategy in~\eqref{eq:decision}.
	\label{th:oacMVstable}
\end{theorem}
\begin{proof}
	Similar to the proof of Theorem~\ref{thm2}, let $\energyFunction[x] \triangleq x^2$ and we have
	\begin{align}
		&  \expectationOperator[{\energyFunction[{\distance[\indexCoordinate][\indexRound+1]}]|\distance[\indexCoordinate][\indexRound]}][]  - \energyFunction[{\distance[\indexCoordinate][\indexRound]}]  \nonumber\\
		&= -2 \updateRate \Trefresh \: \expectationOperator[{\feedbackEle[\indexRound][\indexCoordinate]|\distance[\indexCoordinate][\indexRound]}][] \distance[\indexCoordinate][\indexRound] + \updateRate^2 \Trefresh^2 \expectationOperator[{(\feedbackEle[\indexRound][\indexCoordinate])^2|~\distance[\indexCoordinate][\indexRound]}][]. \nonumber
	\end{align}
	Based on the distribution of $\feedbackEle[\indexRound][\indexCoordinate]$ for the OAC (MV) case, 
	\begin{align*}
		&\expectationOperator[{\feedbackEle[\indexRound][\indexCoordinate]|\distance[\indexCoordinate][\indexRound]}][] = 
	\phiFunc[{+}][{\distance[\indexCoordinate][\indexRound]}] - \phiFunc[{-}][{\distance[\indexCoordinate][\indexRound]}]=1-
	2\positiveHalfOAC~,\nonumber\\
		& \expectationOperator[{(\feedbackEle[\indexRound][\indexCoordinate])^2|~\distance[\indexCoordinate][\indexRound]}][] = 1~.
	\end{align*}
	Therefore,
	\begin{align*}
		&  \expectationOperator[{\energyFunction[{\distance[\indexCoordinate][\indexRound+1]}]|\distance[\indexCoordinate][\indexRound]}][]  - \energyFunction[{\distance[\indexCoordinate][\indexRound]}]  \\
		&=-2 \updateRate \Trefresh \: \left( \phiFunc[{+}][{\distance[\indexCoordinate][\indexRound]}] -  \phiFunc[{-}][{\distance[\indexCoordinate][\indexRound]}]\right) \distance[\indexCoordinate][\indexRound]+ \updateRate^2 \Trefresh^2~. 
	\end{align*}
	Note that $\positiveHalfOAC<1/2$ for $\distance[\indexCoordinate][\indexRound]>0$ and $\positiveHalfOAC>1/2$ for $\distance[\indexCoordinate][\indexRound]<0$ with the proposed OAC scheme, which means
	\[
	\left(\phiFunc[+][{\distance[\indexCoordinate][\indexRound]}] - \phiFunc[-][{\distance[\indexCoordinate][\indexRound]}]\right) \distance[\indexCoordinate][\indexRound] > 0~,
	\]
	always holds.  Thus, following a similar analysis in the proof of Theorem~\ref{thm2}, we conclude that the system is mean-square globally uniformly ultimately bounded and the ultimate bound is determined by $\updateRate\Trefresh$ and $\stdSensor$.
\end{proof}

\def\goldenbaumLength{L}
\section{Numerical Results}
\label{sec:numerical}
In this section, we first numerically analyze the performance of the scheme for an arbitrary application. Subsequently, we apply it to the UAV waypoint flight control scenario discussed in Section~\ref{sec:system}. 
For all analyses, we assume that there are  $\numberOfEdgeDevices=50$ sensors. 
For comparison, we also generate our results with Goldenbaum's non-coherent \ac{OAC} scheme discussed in \cite{Goldenbaum_2013tcom}. 
In this approach, the power of the transmitted signal is modulated. To this end, a transmitter maps three possible votes, i.e., $-1$, $0$, and $1$, to the symbols $0$, $1$, and $2$, respectively, and multiplies a unimodular random sequence of length $\goldenbaumLength$ with the square of the symbol to be aggregated. The receiver calculates the norm-square of the aggregated sequences and re-scales it with $f(x)=x/\goldenbaumLength-\numberOfEdgeDevices$. It then calculates the sign of scaled value to obtain the \ac{MV}. To make a fair comparison, $\goldenbaumLength$ is set to the nearest integer to
 $2^\numberOfIterations/\numberOfIterations$, and $\numberOfIterations$ sequences are mapped to the subcarriers back-to-back to compute $\numberOfIterations$ \acp{MV}. For instance, for $\numberOfIterations=6$, our scheme computes $6$~MVs by using $64$~resources, i.e., $64$~OFDM subcarriers. Hence, we choose $\goldenbaumLength$ to be $11$ as $2^\numberOfIterations/\numberOfIterations\approx10.67$ and use orthogonal resources to compute $6$~MVs for Goldenbaum's approach. We choose the phase of an element of unimodular sequence uniformly between 0 and $2\pi$.  Note that the relative positions of the sensors to the UAV and the room shape do not alter  the performance of the evaluated schemes due to the power control assumption and sensor model discussed in Section~\ref{sec:system}.


\subsection{CER and PMEPR Results}

\begin{figure}
	\centering
	\includegraphics[width =\figuresize]{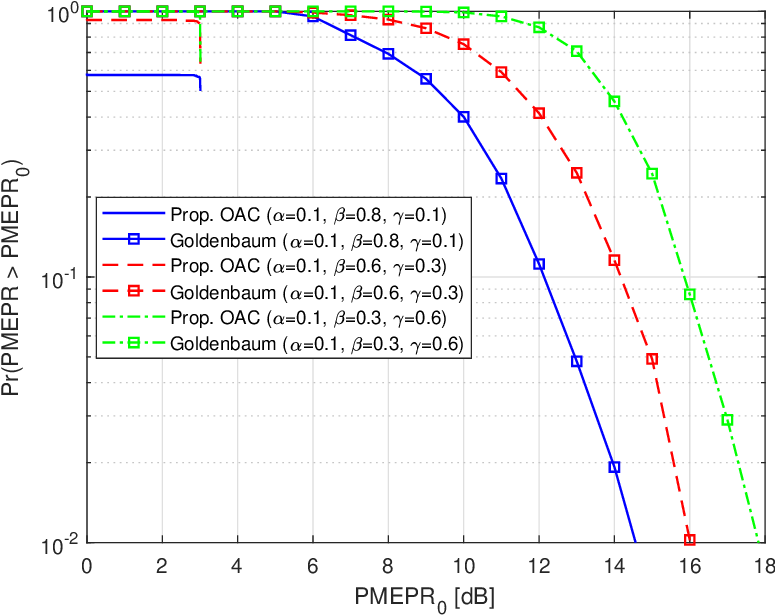}
	\caption{PMEPR distribution.}
	\label{fig:pmepr}
\end{figure}
In \figurename~\ref{fig:pmepr}, we analyze the \ac{PMEPR} distribution of the transmitted signals  for $\numberOfIterations=8$, $\goldenbaumLength=32$,  $\probabilityPlusDecision=0.1$, and $\probabilityNullDecision\in\{0.1,0.3,0.6\}$. As can be seen from \figurename~\ref{fig:pmepr}, the \ac{PMEPR} of a transmitted signal with the proposed scheme is always less than or equal to $3$~dB due to the properties of the \acp{CS}. If there are no absentee votes, the maximum \ac{PMEPR} of the proposed scheme is $0$~dB since a single subcarrier is used for the transmission (see the cases for $\voteVectorED[\indexED]=(1, 1, 1)$ and $\voteVectorED[\indexED]=(1, 1, -1)$ in Example~\ref{ex:sequences}). Hence, for a larger absentee vote probability, the probability of observing $0$~dB PMEPR increases. The combination of sequences that lead to $0$~dB  and $3$~dB PMEPR values and results in the jumps in the PMEPR distribution given in \figurename~\ref{fig:pmepr}. The \ac{PMEPR} characteristics for Goldenbaum's approach are similar to the ones for typical OFDM transmissions and the gap between the proposed scheme and Goldenbaum's approach is considerable large in term of \ac{PMEPR}.

\begin{figure*}
	\centering
	\subfloat[{Flat fading (Rice distribution, infinite $K$-factor, $\probabilityNullDecision=0.1$).}]{\includegraphics[width =\figuresizeS]{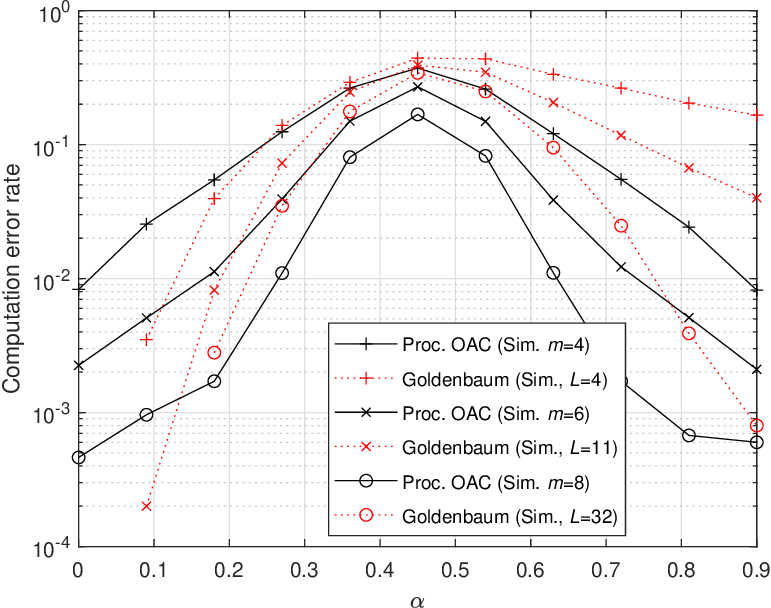}
		\label{subfig:CER01riceFlat}}~~~~~~~~
	\subfloat[{Flat fading (Rice distribution, infinite  $K$-factor, $\probabilityNullDecision=0.6$).}]{\includegraphics[width =\figuresizeS]{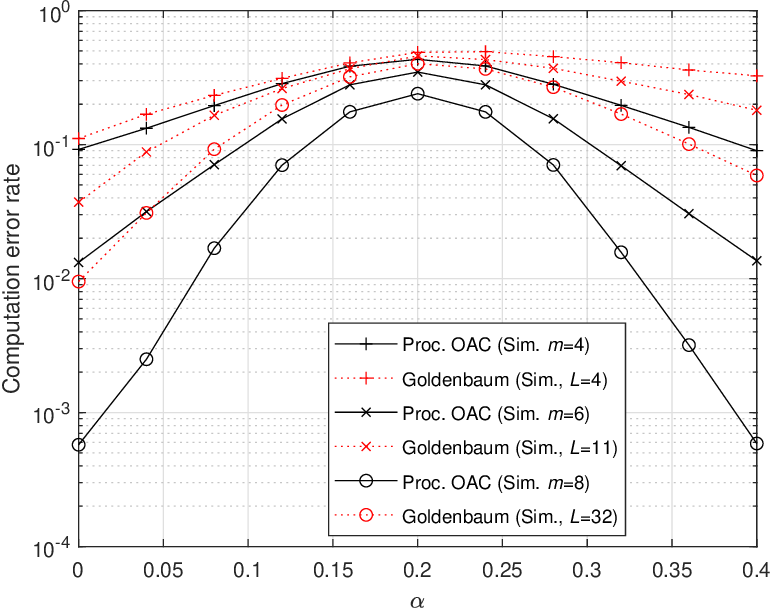}
	\label{subfig:CER06riceFlat}}		
			\\
	\subfloat[{Flat fading   (Rayleigh distribution, $\probabilityNullDecision=0.1$).}]{\includegraphics[width =\figuresizeS]{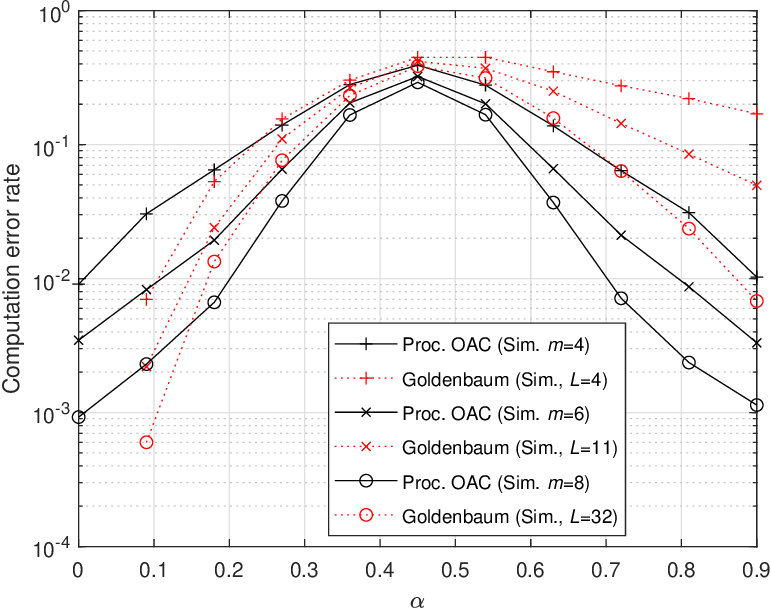}
		\label{subfig:CER01flatFading}}~~~~~~~~
	\subfloat[{Flat fading  (Rayleigh distribution, $\probabilityNullDecision=0.6$).}]{\includegraphics[width =\figuresizeS]{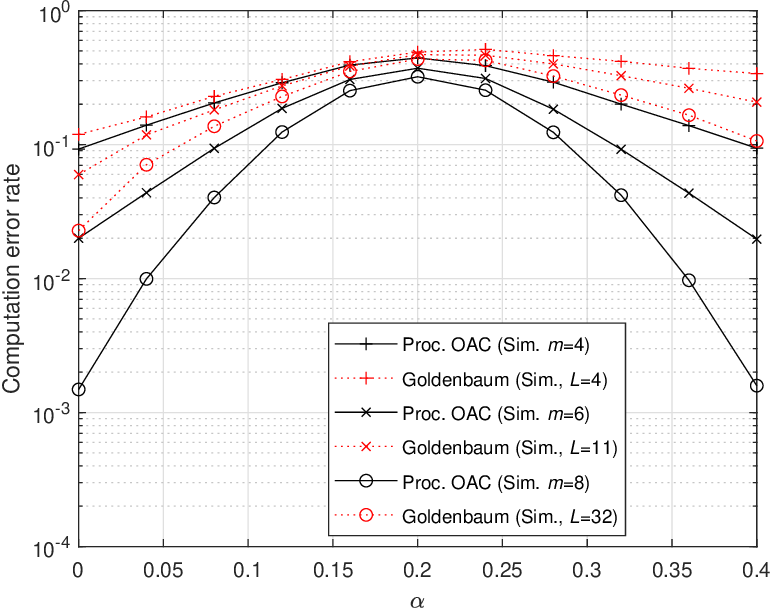}
		\label{subfig:CER06flatFading}}
		\\
	\subfloat[{Frequency-selective fading (Rayleigh distribution, $\probabilityNullDecision=0.1$).}]{\includegraphics[width =\figuresizeS]{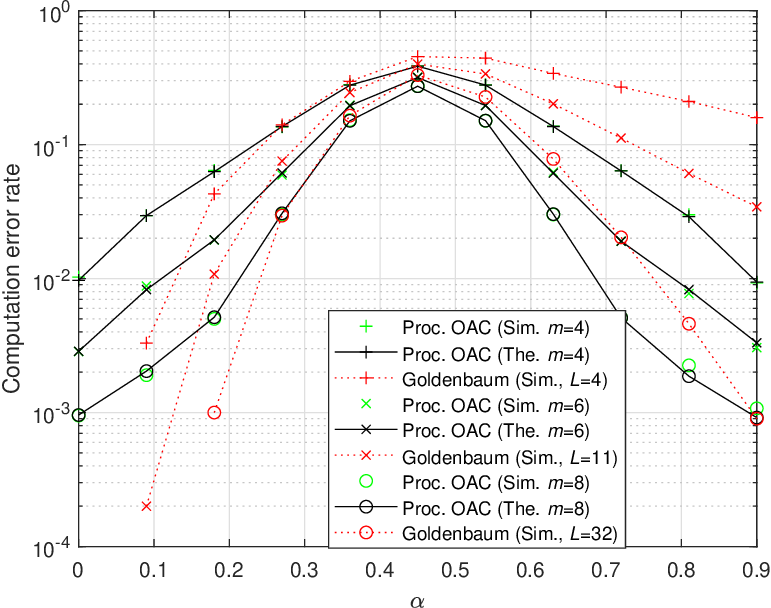}\label{subfig:CER01selectiveFading}}~~~~~~~~
	\subfloat[{Frequency-selective fading (Rayleigh distribution, $\probabilityNullDecision=0.6$).}]{\includegraphics[width =\figuresizeS]{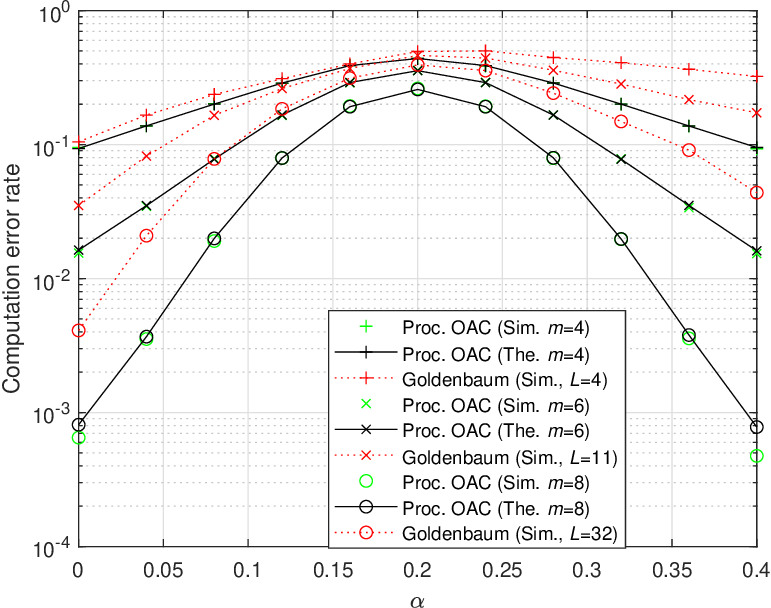}\label{subfig:CER06selectiveFading}}	
	\caption{CER for flat and frequency-selective channels ($\numberOfEdgeDevices=50$ sensors).}
	\label{fig:cerChannel}
\end{figure*}
In \figurename~\ref{fig:cerChannel}, we analyze $\probabilityCER$ defined in \eqref{eq:CERdef} for  $\probabilityNullDecision\in\{0.1,0.6\}$,  $\numberOfIterations=\{4,6,8\}$ for the proposed scheme and $\goldenbaumLength=\{4,11,32\}$ for Goldenbaum's approach by sweeping $\probabilityPlusDecision$ in flat-fading   (i.e., $\channelAtSubcarrier[\indexED,\varMonomial]=\channelAtSubcarrier[\indexED,\varMonomial']\sim\complexGaussian[0][1]$, $\varMonomial\neq\varMonomial'$ for Rayleigh distribution and $\channelAtSubcarrier[\indexED,\varMonomial]=\channelAtSubcarrier[\indexED,\varMonomial']=\constante^{\constantj2\pi\theta_{\indexED,\varMonomial}}$, $\theta_{\indexED,\varMonomial}\sim\uniformDistribution[0][1]$, $\varMonomial\neq\varMonomial'$ for Rice distribution with an infinite Rician $K$-factor) and frequency-selective  (i.e., $\channelAtSubcarrier[\indexED,\varMonomial]\sim\complexGaussian[0][1]$) channels, respectively. For both schemes, as expected, the \ac{CER} improves for a small or a large $\probabilityPlusDecision$ since more sensors vote for $-1$ or $+1$, respectively. The performance in the frequency selective channel is slightly better than the ones in flat-fading channels because of the diversity~gain. Also, both schemes achieve a better \ac{CER} for increasing $\numberOfIterations$ or $\goldenbaumLength$ at the expense of more resource consumption in all channel conditions.
When $\probabilityPlusDecision$ and $\probabilityNullDecision$ are small (as in \figurename~\ref{fig:cerChannel}\subref{subfig:CER01riceFlat},  \figurename~\ref{fig:cerChannel}\subref{subfig:CER01flatFading} and \figurename~\ref{fig:cerChannel}\subref{subfig:CER01selectiveFading}), most of the votes are $-1$. In this case, the norm square of the superposed sequence, calculated at the receiver for Goldenbaum's scheme, becomes less sensitive to the fading coefficients as most of the sensors are silent. On the contrary, if $\probabilityPlusDecision$ gets larger, the norm square of the superposed sequence fluctuates based on the fading coefficients. As a result, the corresponding detector performs poorly compared to the case with smaller $\probabilityPlusDecision$. Thus, Goldenbaum’s scheme exhibits an unbalanced behavior across different $\probabilityPlusDecision$ values. In contrast,  the proposed scheme does not show such unbalanced characteristics while improving the performance in most cases.
 For $\probabilityNullDecision=0.1$, the proposed scheme performs better than Goldenbaum's approach for $\probabilityPlusDecision\in[0.27,0.8]$. When there are more absentee votes, i.e., $\probabilityNullDecision=0.6$, the proposed scheme is superior to Goldenbaum's scheme for all values of $\probabilityPlusDecision$, as can be seen in \figurename~\ref{fig:cerChannel}\subref{subfig:CER06riceFlat},  \figurename~\ref{fig:cerChannel}\subref{subfig:CER06flatFading} and \figurename~\ref{fig:cerChannel}\subref{subfig:CER06selectiveFading}.  Also, the theoretical CERs based on the expression in \eqref{eq:CERpqz} are well-aligned with the simulation results in \figurename~\ref{fig:cerChannel}\subref{subfig:CER01selectiveFading} and \figurename~\ref{fig:cerChannel}\subref{subfig:CER06selectiveFading}.

\if\IEEEsubmission1 \def\figuresize{3.3in} \else \def\figuresize{3.5in} \fi
\begin{figure*}
	\centering
	\subfloat[{$\probabilityPlusDecision=1$, $\probabilityMinusDecision=0$, $\probabilityNullDecision=0$ or $\probabilityPlusDecision=0$, $\probabilityMinusDecision=1$, $\probabilityNullDecision=0$.}]{\includegraphics[width =\figuresize]{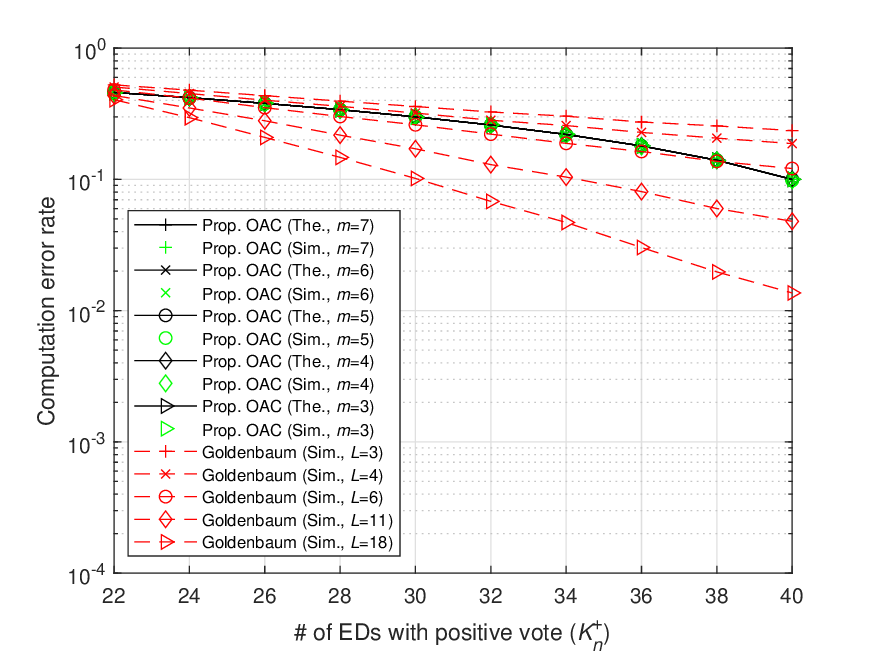}
		\label{subfig:case_p1_p0_z0_errProb}}	
	\subfloat[{$\probabilityPlusDecision=1/2$, $\probabilityMinusDecision=1/2$, $\probabilityNullDecision=0$.}]{\includegraphics[width =\figuresize]{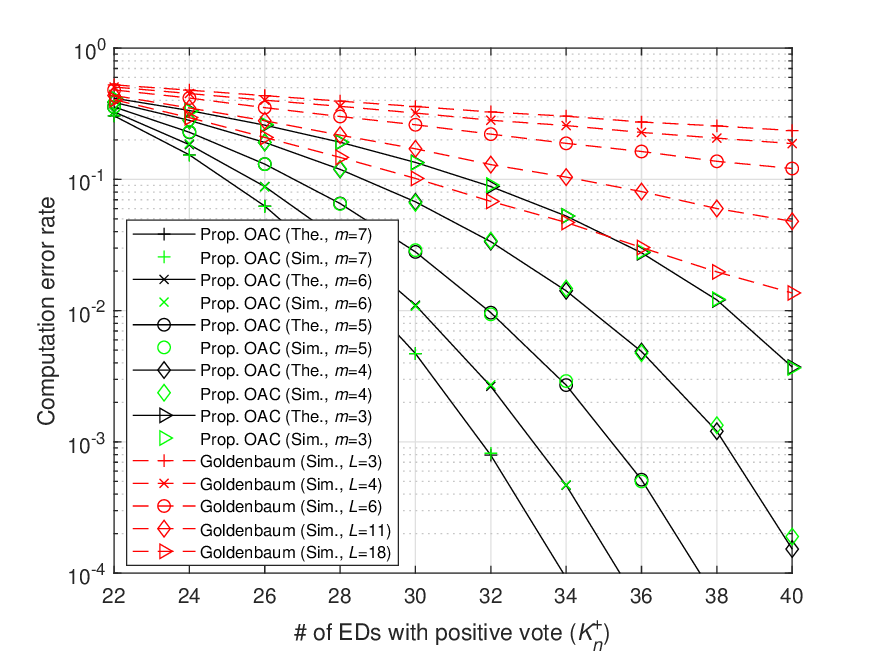}
		\label{subfig:case_p12_p12_z0_errProb}}	\\		
	\subfloat[{$\probabilityPlusDecision=0$, $\probabilityMinusDecision=0$, $\probabilityNullDecision=1$.}]{\includegraphics[width =\figuresize]{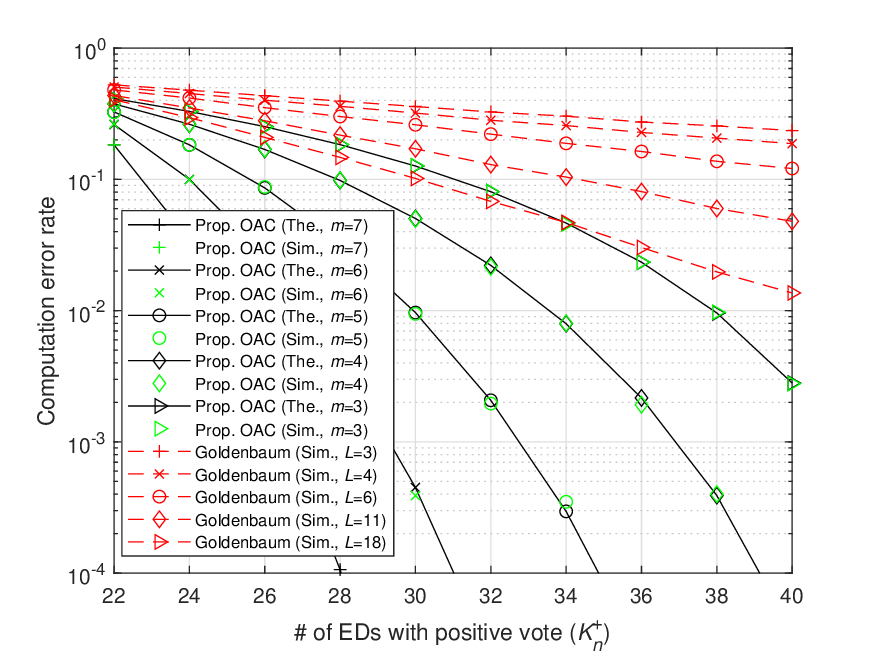}
		\label{subfig:case_p0_p0_z1_errProb}}
	\subfloat[{$\probabilityPlusDecision=1/3$, $\probabilityMinusDecision=1/3$, $\probabilityNullDecision=1/3$.}]{\includegraphics[width =\figuresize]{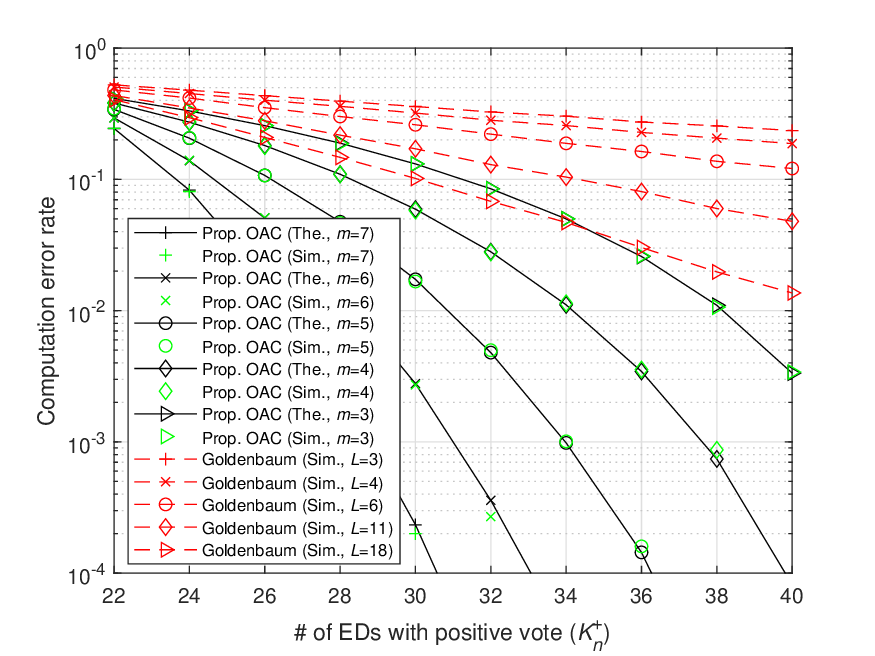}
		\label{subfig:case_p13_p13_z13_errProb}}		
	\caption{$\probabilityErrorPlus[\indexIteration;{\numberOfEDsPlus[\indexIteration], \numberOfEDsMinus[\indexIteration], \numberOfEDsZero[\indexIteration]}]$ for different values of $\probabilityPlusDecision$, $\probabilityMinusDecision$, $\probabilityNullDecision$ ($\numberOfEdgeDevices=50$~EDs, $\numberOfEDsZero[\indexIteration]=10$~EDs, $\SNR=10$~dB, frequency-selective fading channel).}
	\label{fig:cer}
\end{figure*}
In \figurename~\ref{fig:cer}, we evaluate $\computationErrorRate[\indexIteration;{\numberOfEDsPlus[\indexIteration], \numberOfEDsMinus[\indexIteration], \numberOfEDsZero[\indexIteration]}]$ in  Corollary~\ref{cor:cerpqzKpKnKz} in frequency-selective fading channel by increasing $\numberOfEDsPlus[\indexIteration]$ from $22$ to $40$ for $\numberOfEDsZero[\indexIteration]=10$ (i.e., the first case of \eqref{eq:cerKpKnKz}),  $\SNR=10$~dB, $\numberOfIterations\in\{3,4,5,6,7\}$ for the proposed scheme and $\goldenbaumLength\in\{3,4,6,11,18\}$ for the Goldenbaum's approach. In \figurename~\ref{fig:cer}\subref{subfig:case_p1_p0_z0_errProb}, we assume $\probabilityPlusDecision=1$, $\probabilityMinusDecision=0$, and $\probabilityNullDecision=0$ (or $\probabilityPlusDecision=0$, $\probabilityMinusDecision=1$, and $\probabilityNullDecision=0$).  As $\probabilityNullDecision=0$, there are no absentee votes. Also, the same number of sensors activates the same element of the transmitted sequence for all realizations. Since the other elements are not used for the transmission and the energy accumulation is non-coherent, the scheme does not provide any performance gain with increasing $\numberOfIterations$. In \figurename~\ref{fig:cer}\subref{subfig:case_p12_p12_z0_errProb}, we assume that $\probabilityPlusDecision=1/2$, $\probabilityMinusDecision=1/2$, and $\probabilityNullDecision=0$. As compared to the previous case, we observe a significant improvement with increasing $\numberOfIterations$. This is because the randomness enables the votes to accumulate on $2^{\numberOfIterations-1}$ subcarriers, rather than a single resource. Hence, accumulating the energy over multiple subcarriers yields a better estimation of $\metricPlus[\indexIteration]$ and $\metricMinus[\indexIteration]$. A similar result is given in  \figurename~\ref{fig:cer}\subref{subfig:case_p0_p0_z1_errProb} when all the sensors have absentee votes, i.e., {$\probabilityPlusDecision=0$, $\probabilityMinusDecision=0$, and $\probabilityNullDecision=1$.} This is due to the fact that all sensors activate $2^{\numberOfIterations-1}$ elements of the transmitted \ac{CS}. Hence, the \ac{CER} decreases when $\numberOfIterations$ increases. Finally, in \figurename~\ref{fig:cer}\subref{subfig:case_p13_p13_z13_errProb}, we analyze the case for {$\probabilityPlusDecision=1/3$, $\probabilityMinusDecision=1/3$, and $\probabilityNullDecision=1/3$} and show that \ac{CER} performance improves with increasing $\numberOfIterations$. For all cases, the theoretical results  match with the simulations, and the \ac{CER} performance improves with increasing $\numberOfEDsPlus[\indexIteration]$. For Goldenbaum's scheme, each MV is calculated on orthogonal resources. Hence, the corresponding $\computationErrorRate[\indexIteration;{\numberOfEDsPlus[\indexIteration], \numberOfEDsMinus[\indexIteration], \numberOfEDsZero[\indexIteration]}]$ is not a function of $\probabilityPlusDecision$, $\probabilityMinusDecision$, and $\probabilityNullDecision$. We observe that Goldenbaum's scheme performs better for a larger $\goldenbaumLength$. However, since the proposed scheme can exploit the available number of subcarriers much more effectively, it yields notably better performance, as seen in \figurename~\ref{fig:cer}\subref{subfig:case_p12_p12_z0_errProb}-\subref{subfig:case_p13_p13_z13_errProb}.

\subsection{UAV Waypoint Flight Control}
\if\IEEEsubmission1
\begin{figure}
	\centering
	\subfloat[{UAV's trajectory in time.}]{\includegraphics[width =3in]{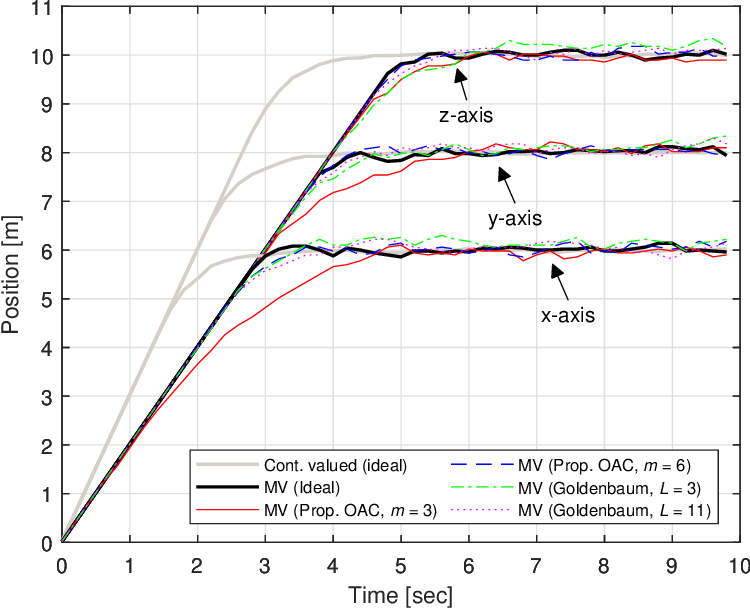}
		\label{subfig:wayPointSingleTime}}	~~
	\subfloat[{UAV's trajectory in space.}]{\includegraphics[width =3.5in]{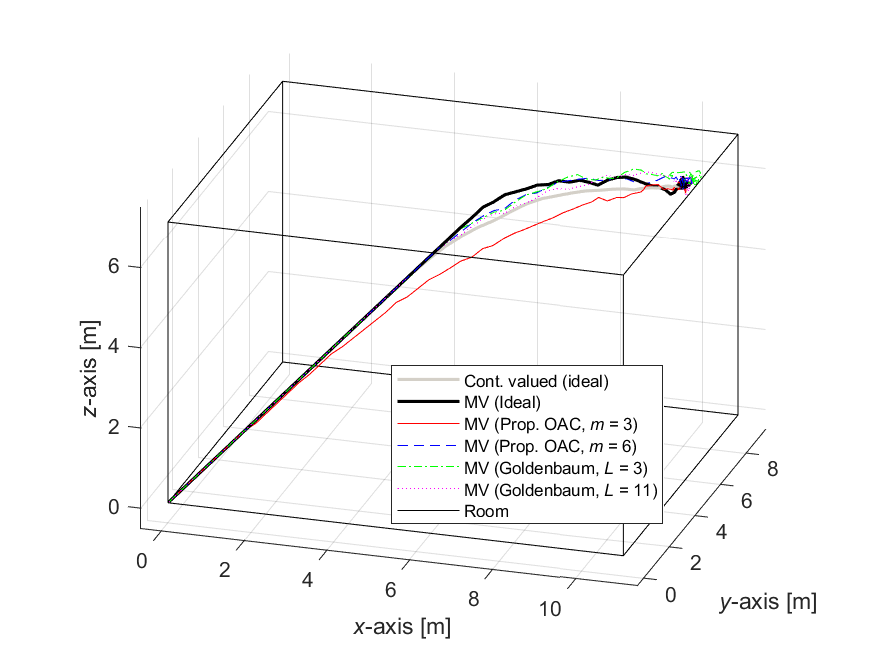}
		\label{subfig:wayPointSingleSpace}}
	\caption{UAV's trajectory with a single point of interest. The initial position is $(0,0,0)$ and the target position is $(10,8,6)$.}
	\label{fig:waypointSingle}
\end{figure}
\begin{figure}
	\centering
	\subfloat[UAV's trajectory in time.]{\includegraphics[width =3in]{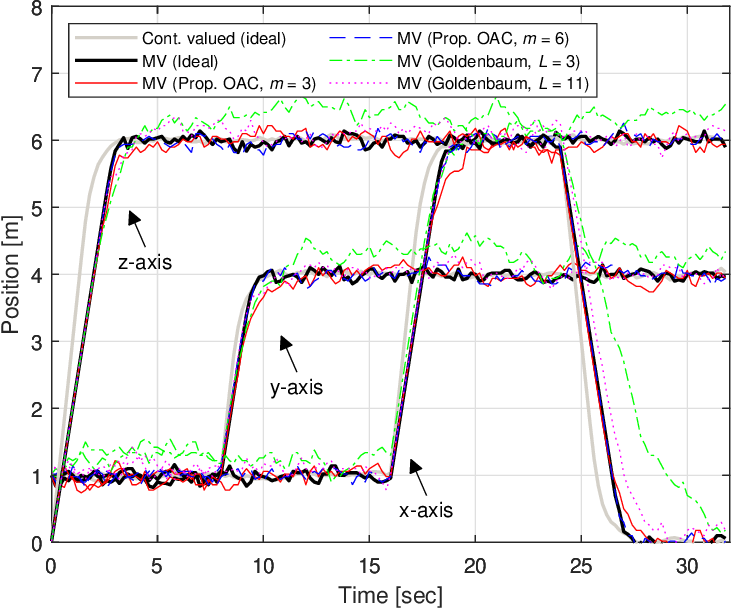}
		\label{subfig:wayPointTime}}~~
	\subfloat[{UAV's trajectory in space.}]{\includegraphics[width =3.5in]{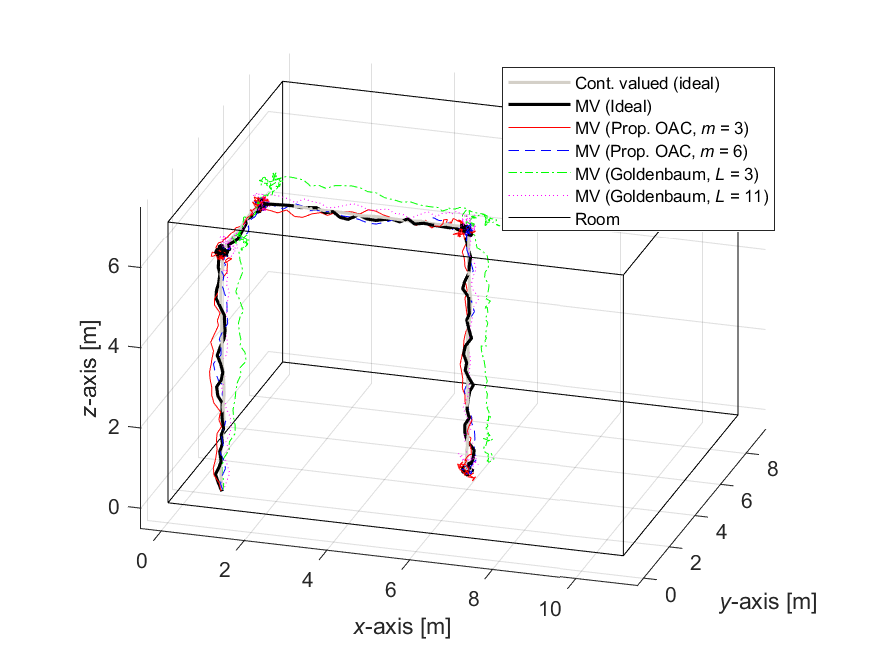}
		\label{subfig:waypointSpace}}		
	\caption{UAV's trajectory with multiple points of interest, i.e., $(1,1,6)$, $(1,4,6)$, $(6,4,6)$, and $(6,4,0)$. The initial position  is $(1,1,0)$.}
	\label{fig:waypoint}
\end{figure}
\else
\begin{figure}
	\centering
	\subfloat[{UAV's trajectory in time.}]{\includegraphics[width =3in]{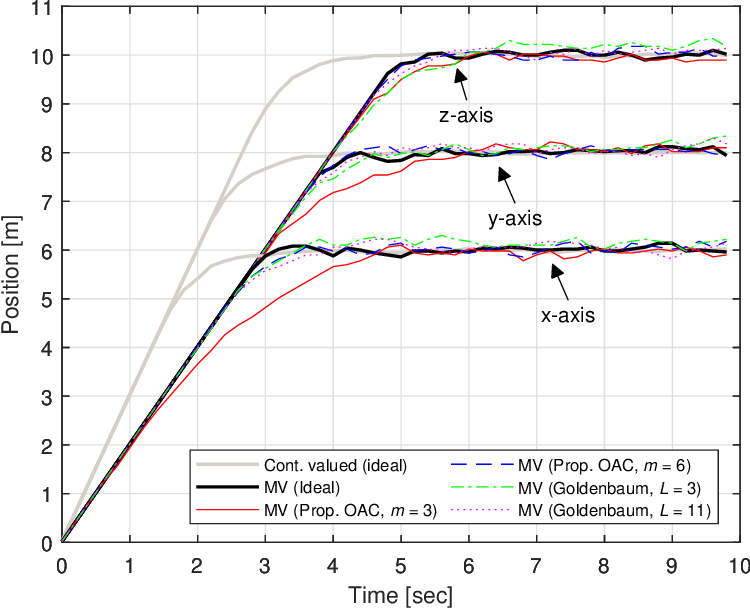}
		\label{subfig:wayPointSingleTime}}	\\
	\subfloat[{UAV's trajectory in space.}]{\includegraphics[width =3.5in]{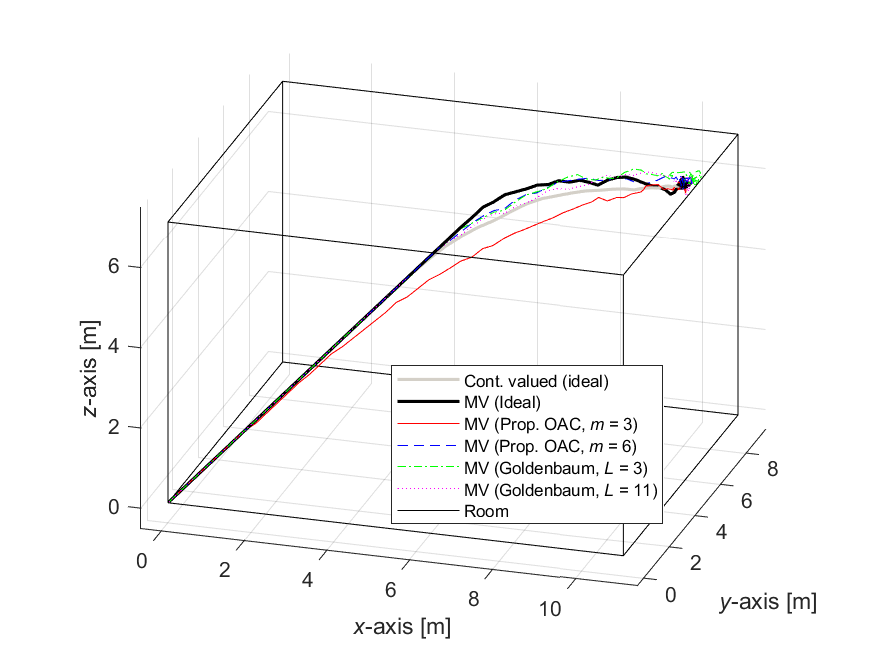}
		\label{subfig:wayPointSingleSpace}}
	\caption{UAV's trajectory with a single point of interest. The initial position is $(0,0,0)$ and the target position is $(10,8,6)$.}
	\label{fig:waypointSingle}
\end{figure}
\begin{figure}
	\centering
	\subfloat[UAV's trajectory in time.]{\includegraphics[width =3in]{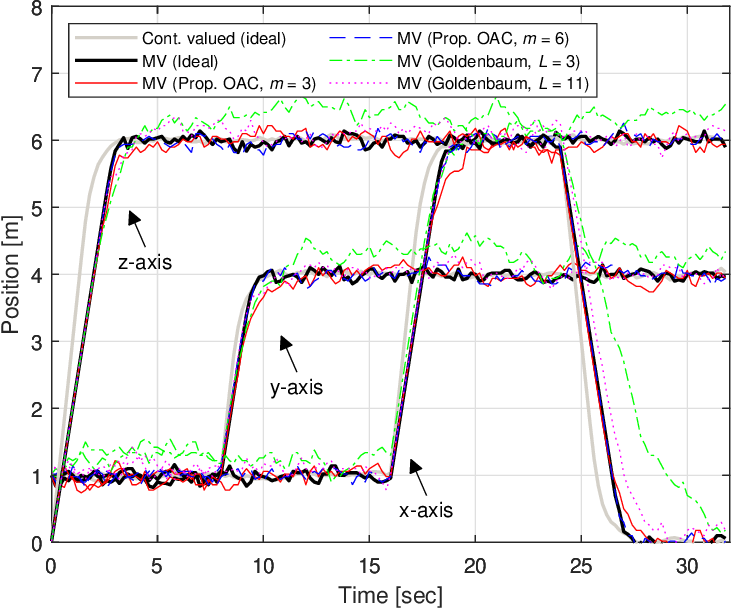}
		\label{subfig:wayPointTime}}\\	
	\subfloat[{UAV's trajectory in space.}]{\includegraphics[width =3.5in]{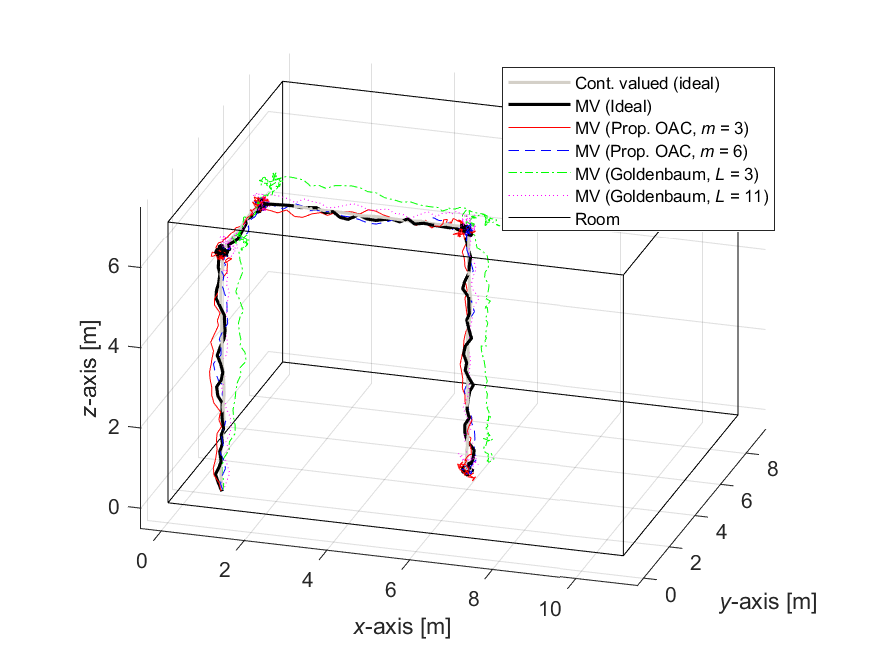}
		\label{subfig:waypointSpace}}		
	\caption{UAV's trajectory with multiple points of interest, i.e., $(1,1,6)$, $(1,4,6)$, $(6,4,6)$, and $(6,4,0)$. The initial position  is $(1,1,0)$.}
	\label{fig:waypoint}
\end{figure}
\fi
 In \figurename~\ref{fig:waypointSingle} and  \figurename~\ref{fig:waypoint},  we consider the UAV  waypoint flight control scenario discussed in Section~\ref{sec:system} for $\numberOfEdgeDevices=50$ sensors.  We assume $\Trefresh=10$~ms, $\updateRate=2$, $\maximumVelocity = 3$~m/s, $\varianceSensor=2$, and $\SNR=10$~dB. We  provide the trajectory of the \ac{UAV} in time and space.  We consider two cases. In the first case, there is only one point of interest $(\locationTargetEle[1],\locationTargetEle[2],\locationTargetEle[3])=(10,8,6)$ and the initial position of the \ac{UAV} is $(0,0,0)$. In the second case, the points of interest are $(1,1,6)$, $(1,4,6)$, $(6,4,6)$, and $(6,4,0)$, where the initial position of the \ac{UAV} is $(1,1,0)$. We compare the  proposed scheme for $\numberOfIterations\in\{3,6\}$ with both continuous and \ac{MV}-based feedback in an ideal communication channel (i.e., no error due to the communication) and Goldenbaum's approach for $\goldenbaumLength\in\{3,11\}$.
   As can be seen from \figurename~\ref{fig:waypointSingle}\subref{subfig:wayPointSingleTime}, for the continuous-valued feedback, the \ac{UAV} reaches its position faster than any \ac{MV}-based approach. This is because the velocity increment is limited by the step size for \ac{MV}-based feedback in our setup. Hence, as can be seen from \figurename~\ref{fig:waypointSingle}\subref{subfig:wayPointSingleSpace}, the \ac{UAV}'s trajectory in space is slightly bent. Since the proposed scheme is also based on the \ac{MV} computation, its characteristics are similar to the one with MV computation in an ideal channel. Since the \ac{CER} with $\numberOfIterations=6$ is lower than the one with  $\numberOfIterations=3$, the proposed scheme for $\numberOfIterations=6$ performs better and its characteristics are similar to the ideal \ac{MV}-based feedback. Goldebaum's approach has similar characteristics to the proposed scheme in terms of trajectory. However, the variation of the UAV position is considerably large when the UAV reaches its final point, as can be seen in \figurename~\ref{fig:waypointSingle}\subref{subfig:wayPointSingleSpace}, in particular for $\goldenbaumLength=3$.  
    The position of the \ac{UAV} in time and space for multiple points of interest is given in \figurename~\ref{fig:waypoint}\subref{subfig:wayPointTime} and \figurename~\ref{fig:waypoint}\subref{subfig:waypointSpace}, respectively. The proposed scheme for $\numberOfIterations=6$ performs similarly to the one with the  \acp{MV} in ideal communications, and increasing $\numberOfIterations$ leads to a more stable trajectory. For $\goldenbaumLength=3$, the trajectory is less stable for Goldenbaum's method. However, its performance improves for a larger $\goldenbaumLength$ for Goldenbaum's approach.

\section{Concluding Remarks}
\label{sec:conclusion}
In this study, we modulate the amplitude of the \ac{CS} based on Theorem~\ref{th:reduced} to develop a new non-coherent \ac{OAC} scheme for \ac{MV} computation. We show that the proposed scheme reduces the \ac{CER} via bandwidth expansion in flat-fading and frequency-selective fading channel conditions while maintaining the \ac{PMEPR} of the transmitted OFDM signals to be less than or equal to $3$~dB. In this work, we derive the theoretical \ac{CER} and provide the convergence analyses for a control scenario. We show that the proposed scheme results in a lower \ac{CER} than Goldenbaum's method while providing a significant \ac{PMEPR} gain. 
Finally, we demonstrate its applicability to a flight control scenario. The proposed scheme with a larger length of sequences performs similarly to the case where \ac{MV} without OAC. The proposed approach can also be utilized in other applications, such as wireless federated learning or distributed optimization over wireless networks using \ac{MV} computation, to address the congestion problems in band-limited wireless channels.

\appendices

\section{Proof of Lemma~\ref{lemma:exp}}
\label{prof:lemma:exp}
We first need the following proposition:
\begin{proposition}
	The following identities hold:
	\begin{align}
		\sum_{\substack{\forall\seqx\in\integers^\numberOfIterations_2\\\monomialAmp[{\permutationMono[{\indexIteration}]}]=1}}
		\constante^{2\funcfForFinalAmplitudeED[\indexED](\seqx)}=\constante^{2\scaleEexpAtRound[\indexED,\indexIteration][\indexRound]}\sum_{\substack{\forall\seqx\in\integers^\numberOfIterations_2\\\monomialAmp[{\permutationMono[{\indexIteration}]}]=0}}
		\constante^{2\funcfForFinalAmplitudeED[\indexED](\seqx)}=\frac{\constante^{2\scaleEexpAtRound[\indexED,\indexIteration][\indexRound]}}{1+\constante^{2\scaleEexpAtRound[\indexED,\indexIteration][\indexRound]}}2^\numberOfIterations~.\nonumber
	\end{align}
	\label{pro:sum}
\end{proposition}
\begin{proof}
	The first identity is because $\constante^{2\scaleEexpAtRound[\indexED,\indexIteration][\indexRound]\monomialAmp[{\permutationMono[{\indexIteration}]}]}=1$ for $\monomialAmp[{\permutationMono[{\indexIteration}]}]=0$. 
	Under \eqref{eq:normalizationCoef}, $\norm{\transmittedSeq[\indexED]}_2^2=2^\numberOfIterations$ holds. Hence,
	\begin{align}
		\norm{\transmittedSeq[\indexED]}_2^2=&\sum_{\forall\seqx\in\integers^\numberOfIterations_2}
		\constante^{2\funcfForFinalAmplitudeED[\indexED](\seqx)}=
		\sum_{\substack{\forall\seqx\in\integers^\numberOfIterations_2\\\monomialAmp[{\permutationMono[{\indexIteration}]}]=0}} 
		\constante^{2\funcfForFinalAmplitudeED[\indexED](\seqx)}+\sum_{\substack{\forall\seqx\in\integers^\numberOfIterations_2\\\monomialAmp[{\permutationMono[{\indexIteration}]}]=1}} 
		\constante^{2\funcfForFinalAmplitudeED[\indexED](\seqx)}\nonumber \\=&\sum_{\substack{\forall\seqx\in\integers^\numberOfIterations_2\\\monomialAmp[{\permutationMono[{\indexIteration}]}]=0}}
		\constante^{2\funcfForFinalAmplitudeED[\indexED](\seqx)}+\constante^{2\scaleEexpAtRound[\indexED,\indexIteration][\indexRound]}\sum_{\substack{\forall\seqx\in\integers^\numberOfIterations_2\\\monomialAmp[{\permutationMono[{\indexIteration}]}]=0}}\constante^{2\funcfForFinalAmplitudeED[\indexED](\seqx)} = 2^\numberOfIterations~.\nonumber
	\end{align}
\end{proof}

\begin{proof}[Proof of Lemma~\ref{lemma:exp}]
	By using Proposition~\ref{pro:sum}, we can calculate $\expectationOperator[{\metricPlus[\indexIteration]}][]$ and $\expectationOperator[{\metricMinus[\indexIteration]}][]$ as 
	\begin{align}
		&\expectationOperator[{\metricPlus[\indexIteration]}][]
		=\sum_{\substack{\forall\seqx\in\integers^\numberOfIterations_2\\\monomialAmp[{\permutationMono[{\indexIteration}]}]=1}} \expectationOperator[{\left|\sum_{\indexED=1}^{\numberOfEdgeDevices}	\channelAtSubcarrier[\indexED,\funcEnum(\seqx)] 	\constante^{\funcfForFinalAmplitudeED[\indexED](\seqx)}	\constante^{\constantj\frac{2\pi}{\numberOfPointsForPSK}\funcfForFinalPhaseED[\indexED](\seqx)}+\noiseAtSubcarrier[\funcEnum(\seqx)]\right|^2}][]
		\nonumber\\
		&=\sum_{\indexED=1}^{\numberOfEdgeDevices} \sum_{\substack{\forall\seqx\in\integers^\numberOfIterations_2\\\monomialAmp[{\permutationMono[{\indexIteration}]}]=1}}
		\constante^{2\funcfForFinalAmplitudeED[\indexED](\varMonomial)}+2^{\numberOfIterations-1}\noiseVariance\nonumber\\
		&=\sum_{\indexED=1}^{\numberOfEdgeDevices}
		\frac{\constante^{2\scaleEexpAtRound[\indexED,\indexIteration][\indexRound]}}{1+\constante^{2\scaleEexpAtRound[\indexED,\indexIteration][\indexRound]}}2^\numberOfIterations+2^{\numberOfIterations-1}\noiseVariance\nonumber\\
		&=2^\numberOfIterations\left(\frac{\constante^{2\scalingParameters}}{1+\constante^{2\scalingParameters}}\numberOfEDsPlus[\indexIteration]+\frac{1}{2}\numberOfEDsZero[\indexIteration]+\frac{\constante^{-2\scalingParameters}}{1+\constante^{-2\scalingParameters}}\numberOfEDsMinus[\indexIteration]+\frac{1}{2}\noiseVariance	\right)\nonumber~,
	\end{align}
	and
	\begin{align}
		&\expectationOperator[{\metricMinus[\indexIteration]}][]
		=\sum_{\substack{\forall\seqx\in\integers^\numberOfIterations_2\\\monomialAmp[{\permutationMono[{\indexIteration}]}]=0}} \expectationOperator[{\left|\sum_{\indexED=1}^{\numberOfEdgeDevices}	\channelAtSubcarrier[\indexED,\funcEnum(\seqx)] 	\constante^{\funcfForFinalAmplitudeED[\indexED](\seqx)}	\constante^{\constantj\frac{2\pi}{\numberOfPointsForPSK}\funcfForFinalPhaseED[\indexED](\seqx)}+\noiseAtSubcarrier[\funcEnum(\seqx)]\right|^2}][]
		\nonumber\\
		&=\sum_{\indexED=1}^{\numberOfEdgeDevices} \sum_{\substack{\forall\seqx\in\integers^\numberOfIterations_2\\\monomialAmp[{\permutationMono[{\indexIteration}]}]=0}}
		\constante^{2\funcfForFinalAmplitudeED[\indexED](\varMonomial)}+2^{\numberOfIterations-1}\noiseVariance\nonumber\\
		&=\sum_{\indexED=1}^{\numberOfEdgeDevices}
		\frac{1}{1+\constante^{2\scaleEexpAtRound[\indexED,\indexIteration][\indexRound]}}2^\numberOfIterations+2^{\numberOfIterations-1}\noiseVariance\nonumber\\
		&=2^\numberOfIterations\left(\frac{1}{1+\constante^{2\scalingParameters}}\numberOfEDsPlus[\indexIteration]+\frac{1}{2}\numberOfEDsZero[\indexIteration]+\frac{1}{1+\constante^{-2\scalingParameters}}\numberOfEDsMinus[\indexIteration]+\frac{1}{2}\noiseVariance	\right)\nonumber~.\end{align}
\end{proof}

\section{Proof of Lemma~\ref{lemma:errProbGivenVote}}
\label{prof:lemma:errProbGivenVote}
\begin{proof}
	For a given $\voteAll$, $|\receivedSeqEle[{\funcEnum(\seqx)}]|^2$ is an exponential random variable with the mean 
	$\rate[{\funcEnum(\seqx)}]^{-1}=\sum_{\indexED=1}^{\numberOfEdgeDevices}
	\constante^{2\funcfForFinalAmplitudeED[\indexED](\seqx)}+\noiseVariance	$ since $\receivedSeqEle[{\funcEnum(\seqx)}]$ is a zero-mean symmetric complex Gaussian distribution in Rayleigh fading channel. Thus, the characteristic function for $|\receivedSeqEle[{\funcEnum(\seqx)}]|^2$ can be calculated as $({1-\constantj\integralVar\rate[{\funcEnum(\seqx)}]^{-1}})^{-1}$, i.e., the Fourier transform of its \ac{PDF}.  
	
	The sum of independent random variables is equal to the convolutions of their \acp{PDF}. Hence, by using the convolution theorem, the characteristic functions of $\metricPlus[\indexIteration]$ and $\metricMinus[\indexIteration]$ can be written as the product of the  characteristic functions of the corresponding exponential random variables as in $\charFcnPlus[\integralVar]$ and $\charFcnMinus[\integralVar]$, respectively. Similarly, the characteristic function of ${\metricPlus[\indexIteration]-\metricMinus[\indexIteration]}$ is equal to $\charFcnPlus[\integralVar]\charFcnMinusConj[\integralVar]$.
	
	Based on the inversion formula given in \cite{waller_1995inversionCDF}, the \ac{CDF} of ${\metricPlus[\indexIteration]-\metricMinus[\indexIteration]}$ can be obtained from its characteristic function as
	\begin{align}
		\CDF[{\metricPlus[\indexIteration]-\metricMinus[\indexIteration]}][][x;\voteAll]=\frac{1}{2}-\int_{-\infty}^{\infty}\frac{\charFcnPlus[\integralVar]\charFcnMinusConj[\integralVar]}{2\pi\constantj\integralVar} \constante^{-\constantj\integralVar\CDFvariable} d\integralVar~.
	\end{align}
\end{proof}

\bibliographystyle{IEEEtran}
\bibliography{references}

\end{document}

%% file: variables.tex
\def\figuresize{3.5in}
\def\figuresizeS{3.0in}
\def\complexNumbers{\mathbb{C}}

\def\integersPositive{\mathbb{Z}^{+}}

\def\functionSpace[#1]{\mathcal{F}(#1)}
\def\realNumbers{\mathbb{R}}
\def\integers{\mathbb{Z}}
\def\constante{{\rm e}}
\def\constantj{{\rm j}}
\def\expectationOperator[#1][#2]{\mathbb{E}\left[#1\right]}
\def\uniformDistribution[#1][#2]{{\mathcal{U}_{[#1,#2]}}}
\def\traceOperator[#1]{{\mathrm{tr}}\{#1\}}
\def\identityMatrix[#1]{\mathrm{\textbf{I}}_{#1}}
\def\zeroVector[#1]{{ {{\mathbf{0}}}}_{#1}}
\def\oneVector[#1]{{ {\mathbf{1}}}_{#1}}
\def\exponentialIntegral[#1]{\mathrm{Ei}(#1)}
\def\functionArbitrary[#1]{f_{#1}}
\def\functionArbitraryEstimate[#1]{\hat{f}_{#1}}

\def\clamp[#1][#2]{\text{clamp}_{#2}\left(#1\right)}
\def\signNormal[#1]{\text{sign}\left(#1\right)}
\def\diagOperation[#1]{\text{diag}\left\{#1\right\}}

\def\probability[#1]{\mathrm{Pr}({#1})}
\def\complexGaussian[#1][#2]{\mathcal{CN}({#1,#2})}
\def\gaussian[#1][#2]{\mathcal{N}({#1,#2})}

\def\normalPDF[#1]{\phi\left(#1\right)}

\def\CDF[#1][#2][#3]{F_{#1}^{#2}\left({#3}\right)}

\def\PDF[#1][#2][#3]{f_{#1}^{#2}\left({#3}\right)}

\def\numberOfEdgeDevices{K}
\def\indexED{k}

\def\voteVectorEDEle[#1][#2]{v_{#1,#2}^{(\indexRound)}}
\def\voteVectorED[#1]{{\textit{\textbf{v}}}_{#1}^{(\indexRound)}}
\def\voteVectorAcrossED[#1]{\textit{\textbf{u}}_{#1}^{(\indexRound)}}
\def\voteAll{\textit{{\textbf{V}}}^{(\indexRound)}}
\def\voteAllWithout{\dot{\textit{{\textbf{V}}}}}

\def\majorityVoteEle[#1]{w_{#1}^{(\indexRound)}}
\def\majorityVoteDetectedEle[#1]{\hat{w}_{#1}^{(\indexRound)}}

\def\varMonomial{{i_{\rm d}}}
\def\monomial[#1]{x_{#1}}
\def\monomialAmp[#1]{y_{#1}}

\def\lengthGaGb{L}

\def\scalingParameters{\xi}
\def\scaleEexp[#1]{a_{#1}}
\def\scaleEexpAtRound[#1][#2]{a_{#1}^{(#2)}}
\def\angleexpAll[#1]{b_{#1}}
\def\angleexpAllAtRound[#1][#2]{b_{#1}^{(#2)}}

\def\permutationMonoD[#1]{{\hat{\pi}_{#1}}}

\def\eleGa[#1]{{a}_{#1}}
\def\eleGb[#1]{{b}_{#1}}
\def\eleGc[#1]{{c}_{#1}}
\def\eleGt[#1]{{t}_{#1}}

\def\angleexpAllBitD[#1]{\hat{k}_{#1}}

\def\permutationMonoD[#1]{{\hat{\pi}_{#1}}}

\def\eleGa[#1]{{a}_{#1}}
\def\eleGb[#1]{{b}_{#1}}
\def\eleGc[#1]{{c}_{#1}}
\def\eleGt[#1]{{t}_{#1}}

\def\apac[#1][#2]{\rho_{#1}(#2)}
\def\apacPositive[#1][#2]{\rho^{+}_{#1}(#2)}
\def\binaryAsignment[#1][#2]{b_{#1}^{(#2)}}
\def\pmfBinomial[#1][#2]{{p}_{\rm b}(#1;#2)}

\def\eleSeqf[#1]{{f}_{#1}}
\def\eleSeqg[#1]{{g}_{#1}}
\def\eleSeqcf[#1]{{c}_{f,#1}}
\def\eleSeqcg[#1]{{c}_{g,#1}}

\def\scaleA[#1]{\alpha_{#1}}
\def\scaleB[#1]{\beta_{#1}}
\def\permutationShift[#1]{{\psi_{#1}}}

\def\permutationMono[#1]{{\pi_{#1}}}
\def\permutationMonoPre[#1]{{\pi'_{#1}}}

\def\seqPermutationCompShift{\bm{\pi}}

\def\vecArrangement[#1]{\textbf{b}_{#1}}

\def\seqGa{\textit{\textbf{a}}}
\def\seqGb{\textit{\textbf{b}}}
\def\seqGaIt[#1]{\textit{\textbf{a}}^{(#1)}}
\def\seqGbIt[#1]{\textit{\textbf{b}}^{(#1)}}

\def\seqGf[#1]{\textit{\textbf{f}}_{#1}}
\def\seqGg[#1]{\textit{\textbf{g}}_{#1}}
\def\seqGfdot[#1]{\bar{\textit{\textbf{f}}}_{#1}}
\def\seqGgdot[#1]{\bar{\textit{\textbf{g}}}_{#1}}

\def\OFDMinTime[#1][#2]{s_{#1}(#2)}
\def\seqx{\textit{\textbf{x}}}
\def\seqGaP{A}

\def\seqGaItP[#1]{{A}^{(#1)}}
\def\seqGbItP[#1]{{B}^{(#1)}}

\def\seqSubP[#1]{{H}_{#1}}
\def\seqGfP[#1]{{{F}}_{#1}}
\def\seqGgP[#1]{{{G}}_{#1}}

\def\seqSub[#1]{\textit{\textbf{h}}_{#1}}
\def\lagForCorrelation{k}

\def\functionSpace{\mathcal{F}}
\def\functionf[#1]{p^{(#1)}}
\def\functiong[#1]{q^{(#1)}}

\def\functionfdot[#1]{{p}_{\indexIterationANF}^{(#1)}} 
\def\functiongdot[#1]{{q}_{\indexIterationANF}^{(#1)}} 

\def\funcfForFinalPhase{f_{\rm i}}

\def\funcfForFinalAmplitude{f_{\rm r}}

\def\funcfForFinalPhaseED[#1]{f_{{\rm i},#1}}
\def\funcfForFinalPhaseDecED[#1]{\check{f}_{{\rm i},#1}}
\def\funcfForFinalAmplitudeED[#1]{f_{{\rm r},#1}}
\def\funcfForFinalAmplitudeDecED[#1]{\check{f}_{{\rm r}, #1}}

\def\funcfForANF{f}

\def\funcEnum{{i}}

\def\varMonomial{{i}}

\def\funcGfForANF[#1]{f_{#1}}
\def\funcGgForANF[#1]{g_{#1}}
\def\polySeq[#1][#2]{{#1}(#2)}

\def\setMonomials[#1]{{\mathcal{M}_{#1}}}

\def\phaseOffsetl[#1]{\Delta_{#1}}
\def\phaseOffsetll[#1]{\Delta_{#1}}
\def\amplitudeOffset[#1]{\epsilon'_{#1}}
\def\amplitudeOffsetl[#1]{\epsilon_{#1}}

\def\lengthOfSequence{L}

\def\totalPowerScale{\delta}

\def\indexEleOfSeq{i}
\def\indexIteration{n}

\def\indexFirstOrderMonomial{j}

\def\orderMonomial[#1]{k_{#1}}
\def\coeffientsANF[#1]{c_{#1}}
\def\polyVariable{z}

\def\numberOfIterations{m}
\def\numberOfPointsForPSK{H}

\def\modulationSymbolF[#1]{m_{#1}}
\def\cardinalitySetOfOperators[#1]{{H}_{#1}}

\def\transmittedSeqP{T}
\def\transmittedSeq[#1]{\textit{\textbf{t}}_{#1}^{(\indexRound)}}
\def\transmittedSeqEle[#1]{{t}_{#1}^{(\indexRound)}}

\def\transmittedSeqNoRound[#1]{\textit{\textbf{t}}_{#1}}
\def\transmittedSeqEleNoRound[#1]{{t}_{#1}}

\def\receivedSeqP{R}
\def\receivedSeq[#1]{\textit{\textbf{r}}_{#1}^{(\indexRound)}}
\def\receivedSeqEle[#1]{{r}_{#1}^{(\indexRound)}}

\def\SNR{{\tt{SNR}}}
\def\channelAtSubcarrier[#1]{h_{#1}}
\def\noiseAtSubcarrier[#1]{\omega_{#1}}
\def\transmitPower[#1]{P_{#1}}
\def\noiseVariance{\sigma_{\rm noise}^2}

\def\numberOfEDsPlus[#1]{K^{+}_{#1}}
\def\numberOfEDsMinus[#1]{K^{-}_{#1}}
\def\numberOfEDsZero[#1]{K^{0}_{#1}}
\def\metricPlus[#1]{E^{+}_{#1}}
\def\metricMinus[#1]{E^{-}_{#1}}

\def\computationErrorRate[#1]{{\tt{CER}}(#1)}
\def\probabilityErrorPlus[#1]{{\tt{CER}^{+}}(#1)}
\def\probabilityErrorMinus[#1]{{\tt{CER}^{-}}(#1)}
\def\rate[#1]{\lambda_{#1}}
\def\probabilityCER{{\tt{CER}}}

\def\Trefresh{T_\mathrm{update}}
\def\indexRound{\ell}

\def\updateRate{\mu}
\def\velocityVectorEle[#1][#2]{{u}^{(#1)}_{#2}}
\def\feedbackEle[#1][#2]{{g}^{(#1)}_{#2}}
\def\locationInitialEle[#1]{{c}_{\mathrm{i},#1}}
\def\locationTargetEle[#1]{{c}_{\mathrm{t},#1}}
\def\locationEle[#1][#2]{{{c}}^{({#1})}_{#2}}
\def\locationEstimateEle[#1][#2]{{{\tilde{c}}}^{({#1})}_{#2}}
\def\locationEstimationErrorEle[#1][#2]{{\epsilon}^{(#1)}_{#2}}
\def\varianceSensor{\sigma_{\rm s}^2}
\def\stdSensor{\sigma_{\rm s}}
\def\indexCoordinate{l}
\def\maximumVelocity{u_\text{\rm limit}}
\def\clamp[#1][#2]{\text{clamp}_{#2}\left(#1\right)}

\def\numberOfCorrectDecision{{N^{+}}}
\def\numberOfZeroDecision{{N^{0}}}
\def\numberOfIncorrectDecision{{N^{-}}}

\def\probabilityPlusDecision{\alpha}
\def\probabilityNullDecision{\gamma}
\def\probabilityMinusDecision{\beta}

\def\computationRate{\mathcal{R}}
\def\efficientImprovement{\epsilon}
\def\rateTraditional{r}
\def\integralVar{t}
\def\CDFvariable{x}
\def\charFcn[#1]{\varphi(#1)}
\def\charFcnPlus[#1]{\varphi^{+}_{\indexIteration}(#1)}
\def\charFcnMinus[#1]{\varphi^{-}_{\indexIteration}(#1)}
\def\charFcnMinusConj[#1]{{{\varphi}^{-}_{\indexIteration}(#1)}^*}
\def\distance[#1][#2]{\delta^{(#2)}_{#1}}

\def\normalCDF[#1]{\Phi\left(#1\right)}
\def\renormalCDF[#1]{\Psi(#1)}

\def\phiFunc[#1][#2]{\phi_{#1}({#2})}
\def\anIndexForK{k}

\def\voteVectorEDEleCoordinate[#1][#2]{\bar{c}_{#1}^{(#2)}}

\def\varP{k}
\def\varN{l}
\def\varZ{m}

\def\positiveHalfOAC{\zeta}

%% file: otherDefinitions.tex
\let\norm\undefined 
\DeclarePairedDelimiter\norm{\lVert}{\rVert}
\newcommand\mydots{\hbox to 1em{.\hss.\hss.}}
\DeclarePairedDelimiter\ceil{\lceil}{\rceil}
\DeclarePairedDelimiter\floor{\lfloor}{\rfloor}

\newtheorem{theorem}{Theorem}
\newtheorem{definition}{Definition}
\newtheorem{lemma}{Lemma}
\newtheorem{proposition}{Proposition}
\newtheorem{corollary}{Corollary}
 
\newtheorem{example}{\color{black} Example}

\DeclareMathOperator{\sign}{sign}

\makeatletter
\newif\ifAC@uppercase@first%
\def\Aclp#1{\AC@uppercase@firsttrue\aclp{#1}\AC@uppercase@firstfalse}%
\def\AC@aclp#1{%
	\ifcsname fn@#1@PL\endcsname%
	\ifAC@uppercase@first%
	\expandafter\expandafter\expandafter\MakeUppercase\csname fn@#1@PL\endcsname%
	\else%
	\csname fn@#1@PL\endcsname%
	\fi%
	\else%
	\AC@acl{#1}s%
	\fi%
}%
\def\Acp#1{\AC@uppercase@firsttrue\acp{#1}\AC@uppercase@firstfalse}%
\def\AC@acp#1{%
	\ifcsname fn@#1@PL\endcsname%
	\ifAC@uppercase@first%
	\expandafter\expandafter\expandafter\MakeUppercase\csname fn@#1@PL\endcsname%
	\else%
	\csname fn@#1@PL\endcsname%
	\fi%
	\else%
	\AC@ac{#1}s%
	\fi%
}%
\def\Acfp#1{\AC@uppercase@firsttrue\acfp{#1}\AC@uppercase@firstfalse}%
\def\AC@acfp#1{%
	\ifcsname fn@#1@PL\endcsname%
	\ifAC@uppercase@first%
	\expandafter\expandafter\expandafter\MakeUppercase\csname fn@#1@PL\endcsname%
	\else%
	\csname fn@#1@PL\endcsname%
	\fi%
	\else%
	\AC@acf{#1}s%
	\fi%
}%
\def\Acsp#1{\AC@uppercase@firsttrue\acsp{#1}\AC@uppercase@firstfalse}%
\def\AC@acsp#1{%
	\ifcsname fn@#1@PL\endcsname%
	\ifAC@uppercase@first%
	\expandafter\expandafter\expandafter\MakeUppercase\csname fn@#1@PL\endcsname%
	\else%
	\csname fn@#1@PL\endcsname%
	\fi%
	\else%
	\AC@acs{#1}s%
	\fi%
}%
\edef\AC@uppercase@write{\string\ifAC@uppercase@first\string\expandafter\string\MakeUppercase\string\fi\space}%
\def\AC@acrodef#1[#2]#3{%
	\@bsphack%
	\protected@write\@auxout{}{%
		\string\newacro{#1}[#2]{\AC@uppercase@write #3}%
	}\@esphack%
}%
\def\Acl#1{\AC@uppercase@firsttrue\acl{#1}\AC@uppercase@firstfalse}
\def\Acf#1{\AC@uppercase@firsttrue\acf{#1}\AC@uppercase@firstfalse}
\def\Ac#1{\AC@uppercase@firsttrue\ac{#1}\AC@uppercase@firstfalse}
\def\Acs#1{\AC@uppercase@firsttrue\acs{#1}\AC@uppercase@firstfalse}

%% file: acronyms.tex
\acrodef{WSN}{wireless sensor network}
\acrodef{USRP}{universal software radio peripheral}
\acrodef{SN}{sensor node}
\acrodef{FC}{fusion center}
\acrodef{MAC}{multiple-access channel}
\acrodef{FL}{federated learning}
\acrodef{ED}{edge device}
\acrodef{CS}{compressed sensing}
\acrodef{ES}[BS]{base station}
\acrodef{DCN}{data center network}
\acrodef{RIS}{reconfigurable intelligent surfaces}
\acrodef{IMC}{in-memory computing}
\acrodef{FPGA}{field-programmable gate array}
\acrodef{SDR}{software-defined radio}
\acrodef{PS}{processing system}
\acrodef{SS}{soft synchronization}
\acrodef{IQ}{in-phase/quadrature}
\acrodef{IP}{intellectual property}
\acrodef{DMA}{direct-memory access}
\acrodef{RAM}{random access memory}
\acrodef{CC}{companion computer}
\acrodef{FEE}{function estimation error}
\acrodef{MSK}{minimum-shift keying}
\acrodef{TDMA}{time-domain multiple access}
\acrodef{PLNC}{physical-layer network coding}
\acrodef{UAV}{unmanned aerial vehicle}
\acrodef{LoRa}{Long-Range}
\acrodef{DC}{direct-current}
\acrodef{DAC}{digital-to-analog converter}
\acrodef{ADC}{anlog-to-digital converter}
\acrodef{CS}{complementary sequence}
\acrodef{GCP}{Golay complementary pair}
\acrodef{ANF}{algebraic normal form}
\acrodef{AACF}{aperiodic auto-correlation function}
\acrodef{RM}{Reed-Muller}

\acrodef{PUCCH}{physical uplink control channel}
\acrodef{PRACH}{physical random access channel}

\acrodef{OBO}{output-power back-off}
\acrodef{ACLR}{adjacent-channel-leakage ratio}

\acrodef{LDPC}{low-density parity check}

\acrodef{PDF}{probability density function}
\acrodef{CDF}{cumulative distribution function}

\acrodef{TBMA}{type-based multiple access}

\acrodef{MSFE}{mean-squared function error}
\acrodef{FEE}{function-estimation error}
\acrodef{CER}{computation error rate}
\acrodef{BCER}{block-computation error rate}
\acrodef{CFO}{carrier frequency offset}
\acrodef{TO}{time offset}
\acrodef{PO}{phase offset}
\acrodef{RSSI}{received signal strength  information}

\acrodef{STLC}{space-time line code}
\acrodef{CCI}{co-channel interference}
\acrodef{CSIT}[CSIT]{\ac{CSI} at the transmitter}
\acrodef{CSIR}[CSIR]{\ac{CSI} at the receiver}
\acrodef{MIMO}{multiple-input-multiple-output}
\acrodef{PC}{phase correction}
\acrodef{ZF}{zero-forcing}
\acrodef{ANOVA}{analysis of variance}

\acrodef{PCA}{principal component analysis}
\acrodef{TIG}{Technical Interest Group}

\acrodef{FSK}{frequency-shift keying}
\acrodef{PPM}{pulse-position modulation}
\acrodef{PAM}{pulse-amplitude modulation}

\acrodef{MRC}{maximum-ratio combining}
\acrodef{HP}{hard-coded participation}
\acrodef{HPA}{hard-coded participation with absentees}
\acrodef{SP}{soft-coded participation}
\acrodef{FSK-MV}{\ac{FSK}-based \ac{MV}}
\acrodef{RF}{radio-frequency}
\acrodef{MF}{matched filter}
\acrodef{PPM}{pulse-position modulation}
\acrodef{CSK}{chirp-shift keying}
\acrodef{PPM-MV}[PPM-MV]{\ac{PPM}-based \ac{MV}}
\acrodef{DFT-s-OFDM}{\ac{DFT}-spread \ac{OFDM}}
\acrodef{SC}{single-carrier}
\acrodef{SGD}{stochastic gradient descent}
\acrodef{signSGD}{sign stochastic gradient descent}

\acrodef{SL}{split learning}
\acrodef{SNR}{signal-to-noise ratio}
\acrodef{RMSE}{root-mean-squared error}
\acrodef{OFDM}{orthogonal frequency division multiplexing}
\acrodef{DFT}{discrete Fourier transform}
\acrodef{PSK}{phase-shift keying}
\acrodef{QAM}{quadrature amplitude modulation}
\acrodef{QPSK}{quadrature phase-shift keying}
\acrodef{PMEPR}{peak-to-mean envelope power ratio}
\acrodef{BER}{bit-error ratio}
\acrodef{SNR}{signal-to-noise ratio}
\acrodef{PSD}{power spectral density}
\acrodef{SE}{spectral efficiency}
\acrodef{CP}{cyclic prefix}
\acrodef{AWGN}{additive white Gaussian noise}
\acrodef{CFR}{channel frequency response}
\acrodef{CIR}{channel impulse response}
\acrodef{MMSE}{minimum mean-squared error}
\acrodef{LMMSE}{linear minimum mean-squared error}
\acrodef{BPSK}{binary phase shift keying}
\acrodef{BPSK}{quadrature phase shift keying}
\acrodef{BLER}{block-error rate}
\acrodef{ML}{maximum likelihood}
\acrodef{PHY}{physical layer}
\acrodef{PA}{power amplifier}
\acrodef{IDFT}{inverse discrete Fourier transform}
\acrodef{DoF}{degrees-of-freedom}
\acrodef{IoT}{Internet-of-Things}
\acrodef{FDE}{frequency-domain equalization}
\acrodef{RF}{radio-frequency}
\acrodef{IM}{index modulation}
\acrodef{MF}{matched filter}
\acrodef{PPM}{pulse-position modulation}

\acrodef{MSE}{mean-squared error}
\acrodef{MRT}{maximum-ratio transmission}
\acrodef{ERC}{equal-ratio combining}
\acrodef{BAA}{broadband analog aggregation}
\acrodef{OBDA}{one-bit broadband digital aggregation}
\acrodef{FEEL}{federated edge learning}
\acrodef{FL}{federated learning}
\acrodef{UL}{uplink}
\acrodef{DL}{downlink}
\acrodef{OAC}{over-the-air computation}
\acrodef{TCI}{truncated-channel inversion}
\acrodef{MV}{majority vote}
\acrodef{CNN}{convolution neural network}
\acrodef{ReLU}{rectified-linear unit}
\acrodef{CSI}{channel state information}
\acrodef{PAPR}{peak-to-average power ratio}
\acrodef{SC}{single-carrier}
\acrodef{iid}[IID]{independent and identically distributed}
\acrodef{RMS}{root-mean-square}
\acrodef{4G}{Fourth Generation}
\acrodef{5G}{Fifth Generation}
\acrodef{NR}{New Radio}
\acrodef{LTE}{Long-Term Evolution}
\acrodef{DFT-s-OFDM}{\ac{DFT}-spread \ac{OFDM}}
\acrodef{OFDMA}{orthogonal frequency division multiple access}
\acrodef{HARQ}{hybrid automatic repeat request}
\acrodef{D2D}{Device-to-Device}
\acrodef{NOMA}{non-orthogonal multiple access}
\acrodef{OMA}{orthogonal multiple access}

\acrodef{IMT}{International Mobile Telecommunications}
\acrodef{ITU}{International Telecommunication Union}